\documentclass[11pt]{article}

\usepackage[margin=1in]{geometry}

\usepackage{amssymb, amsfonts}
\usepackage{amsmath}
\usepackage{amsthm}
\usepackage{empheq}
\usepackage[english]{babel}
\usepackage[utf8]{inputenc}
\usepackage{latexsym}
\usepackage{mathrsfs}
\usepackage{bbm}
\usepackage{graphicx}
\usepackage{float}
\usepackage{subcaption}
\usepackage{xcolor}
\usepackage[hidelinks]{hyperref}
\usepackage{accents}
\usepackage{cases}
\usepackage{url}
\usepackage{comment}
\usepackage{cleveref}
\usepackage{algorithm}
\usepackage{algorithmic}
\usepackage{booktabs}
\usepackage{tabularx}
\newtheorem{definition}{Definition}
\newtheorem{theorem}{Theorem}
\newtheorem{remark}{Remark}
\newtheorem{proposition}{Proposition}

\newtheorem{lemma}{Lemma}
\newcommand\EE{\mathbb{E}}
\newcommand\RR{\mathbb{R}}
\newcommand\TT{\mathbb{T}}

\newcommand{\PP}{\mathbb{P}}

\newcommand{\ba}{\begin{array}}
\newcommand{\ea}{\end{array}}

\newcommand{\br}{\mathbb{R}}

\newcommand{\bp}{\mathbb{P}}

\newcommand{\FCal}{\mathcal{F}}

\newcommand{\PCal}{\mathcal{P}}

\newcommand{\WCal}{\mathcal{W}}
\newcommand{\XCal}{\mathcal{X}}

\newcommand{\ZCal}{\mathcal{Z}}

\newcommand{\argmin}{\mathop{\rm argmin}}
\newcommand{\argmax}{\mathop{\rm argmax}}

\newcommand{\bpm}{\begin{pmatrix}}
\newcommand{\epm}{\end{pmatrix}}
\renewcommand{\div}{\textnormal{div}\,}

\title{Connecting GANs, mean-field games, and optimal transport}
\author{Haoyang Cao\thanks{Alan Turing Institute, London, United Kingdom. 
 Email: hcao@turing.ac.uk}
 \and Xin Guo\thanks{Department of Industrial Engineering and Operations Research, University of California, Berkeley, Berkeley, California, USA. Email: xinguo@berkeley.edu} \and Mathieu Lauri\`ere
 \thanks{Department of Operations Research and Financial Engineering, Princeton University, Princeton, New Jersey, USA.
 Email: lauriere@princeton.edu}}
\date{\today}

\begin{document}
\maketitle

\begin{abstract}
Generative adversarial networks (GANs) have enjoyed tremendous success in image generation and processing, and have recently attracted growing interests in financial modelings. This paper analyzes GANs from the perspectives of mean-field games (MFGs) and optimal transport. More specifically, from the game theoretical perspective, GANs are interpreted as MFGs under Pareto Optimality criterion or mean-field controls; from the optimal transport perspective, GANs are to minimize the optimal transport cost indexed by the generator from the known latent distribution to the unknown true distribution of data. 
The MFGs perspective of GANs leads to a GAN-based computational method (MFGANs) to solve MFGs: one neural network for the backward Hamilton-Jacobi-Bellman equation and one neural network for the forward Fokker-Planck equation, with the two neural networks trained in an adversarial way. Numerical experiments demonstrate superior performance of this proposed algorithm, especially in the higher dimensional case, when compared with existing neural network approaches.
\end{abstract}%



\maketitle

%

\section{Introduction}

\paragraph{GANs.}
Generative Adversarial Networks (GANs), introduced in \cite{Goodfellow2014}, belong to the class of generative models. 
The key idea behind GANs is to interpret the process of generative modeling as a competing game between two neural networks: a generator network $G$ and a discriminator network $D$. The generator network $G$ attempts to fool the discriminator network by converting random noise into sample data, while the discriminator network $D$ tries to identify whether the sample is faked or true.

 After being introduced to the machine learning community, the popularity of GANs has grown exponentially with numerous applications. Some of the popular applications include high resolution image generation \cite{denton2015deep,radford2015unsupervised}, image inpainting \cite{yeh2016semantic}, image super-resolution \cite{ledig2016others}, visual manipulation \cite{zhu2016generative}, text-to-image synthesis \cite{reed2016generative}, video generation \cite{vondrick2016generating}, semantic segmentation \cite{luc2016semantic}, and abstract reasoning diagram generation \cite{ghosh2016contextual}. Recently GANs have been used for simulating financial time-series data \cite{Wiese2019}, \cite{Wiese2020}, \cite{Zhang2019}, and for asset pricing models \cite{Chen2019}.

Along with the empirical success of GANs, there is a growing emphasis on the theoretical analysis of GANs. \cite{Berard2020} proposes a novel visualization method for the GANs training process through the gradient vector field of loss functions.
In a deterministic GANs training framework, \cite{Mescheder2018} demonstrates that regularization improved the convergence performance of GANs; \cite{Conforti2020} and \cite{Domingo-Enrich2020} analyze a generic zero-sum minimax game including that of GANs, and connect the mixed Nash equilibrium of the game with the invariant measure of Langevin dynamics. Recently, \cite{cao2020approximation} analyzes convergence of GANs training process by studying the long-term behavior of its continuous time limit, via the invariant measure of associated coupled stochastic differential equations.

\paragraph{MFGs.} The theory of Mean-field games (MFGs), pioneered by Lasry and Lions (2007) and Huang, Malham\'e and Caines (2006), presents a powerful approach to study  stochastic games of a large population with small interactions.  MFG avoids directly analyzing the otherwise-notoriously-difficult $N$-player stochastic games.  The key idea, coming from physics for very large systems of interacting particles, is to approximate the dynamics and the objective function under the notion of population's probability distribution flows, a.k.a., the mean-field information process.  
By assuming players are indistinguishable and interchangeable, and by the aggregation approach and the strong law of large numbers, MFGs focus on a representative player and the mean-field information. The value function of MFGs is then shown to approximate that of the corresponding $N$-player games with an error of order $\frac{1}{\sqrt{N}}$; see for instance  \cite{Carmonaa, Carmona2018}, \cite{Guo2019}, and the references within.

One of the approaches to find the Nash equilibrium in a MFG is the fixed-point approach, characterized by  a system of coupled partial differential equations (PDEs): a backward Hamilton-Jacobi-Bellman (HJB) equation for the value function of the underlying control problem, and a forward Fokker-Planck (FP) equation for the dynamics of the population (see \cite{Lasry2007},  \cite{Bensoussan2013}, \cite{Carmonaa}, and the references therein).  In addition to the PDE approach (see also  Gu\'eant, Lasry, and Lions (2010)), there are approaches based on backward stochastic differential equations (BSDEs) by Buckdahn, Djehiche, Li, and Peng (2009) and Buckdahn, Li, and Peng (2009), the probabilistic approach by Carmona and Delarue (2013, 2014) and Carmona and Lacker (2015).

\paragraph{Optimal transport.} The theory of optimal transport, originated from Monge \cite{monge1781memoire}, 
studies the optimization problem of transporting one given initial distribution to another given terminal distribution so that the transport cost functional is minimized. Deeply rooted in linear programming, many theoretical works on the existence and uniqueness of  an optimal transport plan focus on the duality gap between the primal optimization problem and its dual form. For instance, the Kantorovich-Rubinstein duality in \cite{villani2008optimal} characterizes sufficient conditions for the existence of optimal transport plans (i.e., when there is no duality gap), and \cite{Griessler2018} studies the multi-marginal case.

Martingale optimal transport problem is motivated mostly by problems from finance, starting  with the problem of super-hedging, see for instance \cite{Dolinsky2014},
  \cite{Lim2016}, and \cite{Nutz2019a}. \cite{Beiglbock2017a} establishes a complete duality theory for generic martingale optimal transport problems via a quasi-sure formulation of the dual problem;
and \cite{Guo2017} studies a continuous-time martingale optimal transport problem and establishes the duality theory via the S-topology and the dynamic programming approach.

To compute the optimal transport plan, \cite{Guo2017a} proposes a computational method for martingale optimal transport problems based on linear programming via proper relaxation of the martingale constraint and discretization of the marginal distributions; \cite{Eckstein2018} proposes a deep learning algorithm to solve multi-step, multi-marginal optimal transport problem via its dual form with an appropriate penalty term.

\paragraph{Our work.} 
So far, theories of GANs, MFGs, and optimal transport  have been developed independently. 
We will show in this work that they are tightly connected. In particular, GANs can be understood and analyzed from the perspective of  MFGs and optimal transport. More precisely, 
\begin{itemize}
\item We first show a conceptual connection between GANs and MFGs: MFGs have the structure of GANs, and GANs are MFGs under the Pareto Optimality criterion. This intrinsic connection is transparent for a class of MFGs for which there is a minimax game representation. 
\item We then establish rigorously an analytical connection between a broad class of GANs and optimal transport problems from the latent distribution to the true distribution. 
This representation is explicit for Wasserstein GANs as well as some its variations including relaxed Wasserstein GANs. 
At the core of this connection is the Kantorovich-Rubinstein type duality.

\item We finally propose a GANs-based algorithm (MFGANs) to solve MFGs. This is derived from connecting the variational structure embedded in the PDE system for  MFGs with the minimax game structure for GANs. 
This connection suggests a new neural network based  approach to solve MFGs:
one neural network for the backward HJB equation and the other for the forward FP equation, with the two neural networks trained in an adversarial way.
Our numerical experiments demonstrate the performance of this proposed algorithm when compared with existing neural-network-based approaches, especially in higher dimensional case. 
This idea of developing neural network-based algorithm with incorporation of adversarial training appears promising for more general dynamic systems with variational structures. 
\end{itemize}

\paragraph{Related ML techniques for computing MFGs.}
Most existing computational approaches for solving MFGs adopt traditional numerical schemes, such as finite differences~\cite{MR2679575} or semi-Lagrangian~\cite{MR3148086} schemes. Some exceptions are \cite{CarmonaLauriere_DL,CarmonaLauriere_DL_periodic} and \cite{guo2019learning}. \cite{guo2019learning}
 designs reinforcement learning algorithms with convergence and complexity analysis for learning MFGs, where the cost function of the game as well as the parameters for the underlying dynamics are unknown. \cite{al2018solving,CarmonaLauriere_DL,CarmonaLauriere_DL_periodic} propose deep neural networks approaches for solving MFGs, with a particular Deep-Galerkin-Method architecture, to approximate the density and the value function by neural networks separately. 
 \cite{lin2020apac} proposes a GANs-based algorithm named APAC-Net via a primal-dual formulation associated with the coupled HJB-FP system by \cite{Cirant2018}; their numerical experiments demonstrate  their capability of solving some special forms of high-dimensional MFGs. In contrast, our algorithm takes full advantage of the variational structure of MFGs and train the neural networks in an adversarial fashion. Our numerical experiment demonstrates clear advantage of this variational approach, especially in terms of computational efficiency for high dimensional MFGs.  

\paragraph{Related works on connecting GANs and optimal transport.}

There are earlier studies exploring connections between GANs and optimal transport. For instance, \cite{Arjovsky2017} points out that the Wasserstein distance between true distribution $\bp_r$ and generated distribution $\bp_\theta$ can be seen as the optimal transport cost from $\bp_r$ to $\bp_\theta$. \cite{lei2019geometric} provides a geometric  interpretation of Wasserstein GANs. This connection has also been exploited  to improve the stability and performance of GANs training
in \cite{gulrajani2017improved},  \cite{sanjabi2018convergence},   and \cite{chu2019}. In addition, 
\cite{salimans2018improving} defines a novel optimal transport-type divergence to replace the Jensen-Shannon divergence for GANs training. The work \cite{luise2020generalization} studies the generalization property of GANs if the generator is jointly trained with latent distribution, where optimal transport metric is used for the discriminator. An unbalanced optimal transport problem is solved using GANs in \cite{yang2018scalable}. 
Our work establishes this connection analytically and in a more general framework, to allow for a broader class of generative models without the symmetry condition on the loss function.

{\paragraph{Organization.} This paper is organized as follows. Section \ref{sec:prelim} focuses on the basic mathematics and preliminaries for GANs, MFGs and optimal transport problem. Section \ref{sec:gan2mfg} demonstrates that GANs are MFGs in a collaborative setting. The connection between GANs and optimal transport is established in Section \ref{sec:gan-ot}. Finally,  Section \ref{sec:num} explains the analogy between the MFGs and GANs and provides a GAN-based algorithm with numerical experiments. }

\paragraph{Notations.}
Throughout this paper, the following notations will be adopted, unless otherwise specified.
\begin{itemize}
 \item $\XCal$ denotes a Polish space with metric $l$ and $\ZCal$ denotes another Polish space with metric $l_z$. 
 \item $\mathcal{P}(\XCal)$ denotes the set of all probability distributions over the space $\XCal$. 
 \item $\PCal_D(\br^d)$ denotes the set of probability distributions on $\br^d$ that admit corresponding density functions with respect to Lebesgue distribution.
 \item For $p>0$,
 $\PCal^p(\br^d)=\biggl\{\mu\in\PCal_D(\br^d)\biggl|\int_{\br^d}\|x\|_p^p\mu(dx)<\infty\biggl\},$
 with $\|\cdot\|_p$ the $p$-norm on $\br^d$.
 \item For $p\geq1$ and some arbitrary $x_0\in\XCal$, 
 $L^p(\XCal): =\biggl\{\mu\in\PCal(\XCal)\biggl|\int_\XCal l(x,x_0)^p\mu(dx)<\infty\biggl\}$. Note that this set is independent of the choice of $x_0$.
 \item For any given $\mu\in\PCal(\XCal)$, 
 $L^1(\XCal,\mu): =\biggl\{\psi:\XCal\to\br\biggl|\int_\XCal|\psi(x)|\mu(dx)<\infty\biggl\}.$
\end{itemize}

\section{Preliminaries of GANs, MFGs, and Optimal transport} \label{sec:prelim}

\subsection{GANs} GANs fall into the category of generative models.
The procedure of generative modeling is to approximate an unknown probability distribution $\bp_r$ by constructing a class of suitable parametrized probability distributions $\bp_\theta$. That is, given a latent space $\ZCal$ and a sample space $\XCal$, define a latent variable $Z$ taking values in $\ZCal$ with a fixed probability distribution $\PP_z$ 
and a family of functions $G_\theta: \ZCal\to\XCal$ parametrized by $\theta$. Then $\bp_\theta$ is defined as the probability distribution of $G_\theta(Z)$, denoted by $Law (G_\theta(Z))$. 

The key idea of GANs as generative models is to  introduce a competing neural network, namely a discriminator $D = D_\omega: \XCal \to [0,1]$, parametrized by $\omega$. This discriminator assigns a score between $0$ to $1$ to each sample generated either from the true distribution ${\mathbb P}_r$ or the approximate distribution ${\mathbb P}_{\theta}$. A higher score from the discriminator $D$ indicates that the sample is more likely to be from the true distribution. GANs are trained by optimizing $G$ and $D$ iteratively until $D$ is very good at identifying which samples come from $\PP_r$ but the generator is so good that it can fool $D$ to make it believe that the samples from $\PP_\theta$ come from $\PP_r$. 


\paragraph{GANs as minimax games.}
As explained in the paper \cite{Goodfellow2014}, GANs can be formally expressed as the following minimax game:
\begin{equation} \label{gan-obj}
\begin{aligned}
	\inf_\theta \sup_\omega& \Big\{\EE_{X \sim \bp_r} [\log D_\omega(X)]+ \EE_{Z \sim \bp_z} [\log (1 - D_\omega(G_\theta(Z)))]\Big\}.
	\end{aligned}
\end{equation}
Now,  fixing $G$ and optimizing for $D$ in \eqref{gan-obj}, the optimal discriminator would be 
\[D^*_G(x) = \frac{p_r(x)}{p_r(x) + p_\theta(x)},\]
where $p_r$ and $p_\theta$ are density functions of $\bp_r$ and $\bp_\theta = Law(G(Z))$ respectively, assuming they exist. Plugging this back in~\eqref{gan-obj}, the  GANs minimax game becomes \begin{align*}
&
\begin{aligned}
 \min_G &\left\{\EE_{X \sim \bp_r}\left[\log \frac{p_r(X)}{p_r(X) + p_\theta(X)}\right] + \EE_{Y \sim \bp_\theta} \left[\log \frac{p_\theta(Y)}{p_r(Y) + p_\theta(Y)}\right]\right\} 
\end{aligned}
	\\
	& \hspace{10pt}= -\log 4 + 2 JS(\bp_r, \bp_\theta). \label{gan-obj2}
\end{align*}
That is, training of GANs with an optimal discriminator is equivalent to minimizing the Jensen--Shannon (JS) divergence between $\bp_r$ and $\bp_\theta$. 

Viewing GANs training as an optimization problem between $\bp_r$ and $\bp_\theta$ has led to variants of GANs with different choices of divergences, in order to improve the training  performance: for instance, \cite{nock2017f} uses f-divergence, \cite{srivastava2019bregmn} explores scaled Bregman divergence, \cite{Arjovsky2017} adopts Wasserstein-1 distance, \cite{guo2017relaxed} proposes relaxed Wasserstein divergence, and \cite{salimans2018improving} and \cite{sanjabi2018convergence} utilize the Sinkhorn loss. This viewpoint is instrumental 
to establish the  connection between GANs and optimal transport, studied in Section \ref{sec:gan-ot}.

\paragraph{Equilibrium of GANs training.}
Under a fixed network architecture, the parametrized version of GANs training is to find
\begin{equation} \label{gan-obj-param}
\begin{aligned}
	v_U^{GAN}&=\inf_\theta \sup_\omega L_{GAN}(\theta,\omega),\\
\text{where }L_{GAN}(\theta,\omega)&=\EE_{X \sim \bp_r} [\log D_\omega(X)]+ \EE_{Z \sim \bp_z} [\log (1 - D_\omega(G_\theta(Z)))].
	\end{aligned}
\end{equation}
From a game theory viewpoint, the objective in \eqref{gan-obj-param} is in fact the upper value of the two-player zero-sum game of GANs. Meanwhile, the lower value of the game is given by the following maximin problem,
\begin{equation} \label{gan-obj-maximin}
\begin{aligned}
	v_L^{GAN}=\sup_\omega\inf_\theta L_{GAN}(\theta,\omega).
	\end{aligned}
\end{equation}

Clearly, 
 $v_L^{GAN}\leq v_U^{GAN}.$
Moreover, if there exists a pair of parameters $(\theta^*,\omega^*)$ achieving both \eqref{gan-obj-param} and \eqref{gan-obj-maximin}, then $(\theta^*,\omega^*)$ is a Nash equilibrium of this two-player zero-sum game. As pointed out by Sion's theorem (see \cite{von1959theory} and \cite{sion1958general}), if $L_{GAN}$ is convex in $\theta$ and concave in $\omega$, then there is no duality gap and  $v_L^{GAN}= v_U^{GAN}$. 
It is worth noting that conditions for such an equality  are usually not satisfied in many common GAN models, as stressed out in \cite{zhu2020deconstructing}.




\subsection{MFGs}
\paragraph{GANs and coupled PDE Systems.}
 MFGs are introduced to approximate Nash equilibria in $N$-player stochastic games by  considering a game with infinitely many interchangeable agents. The key idea of MFGs, thanks to the law of large numbers, is to study the interaction between a representative player and the aggregated information of all the opponents instead of focusing on each one of them individually. This aggregated information is often referred to as the mean-field information. Take for instance, a filtered probability space
 $(\Omega,\FCal, \{\FCal_t\}_{t\geq0},\PP)$   which  supports a standard $d$-dimensional Brownian motion $W=\{W_t\}_{t\geq0}$. A class of continuous-time MFGs is to find for the representative player an optimal control $\{\alpha_t\}_{t\geq0}$ with $\alpha_t\in\br^d$ for all $t\geq0$ from an appropriate admissible control set, and to solve for the  following minimization problem: For any $s\in[0,T]$ and $x\in\RR^d$, find
\begin{equation}\label{eq:mfg-eg}\tag{MFG}
 u(s,x)=u(s,x;\{\mu_t\}_{t\in[0,T]})=\inf_{\{\alpha_t\}_{t\geq0}} E\left[\int_s^T f(t,X_t,{\mu_t},\alpha_t)dt \left\vert X_s=x\right.\right],
 \end{equation}
\ \ \ \ \ subject to the state dynamics
\begin{equation*}\label{eq: state-evol}
\begin{aligned}
&dX_t = b(t,X_t,{\mu_t},\alpha_t)dt+\sigma dW_t,\,X_0\sim\mu^0,\\
& \mu^0(dx)=m^0(x)dx.
 \end{aligned}
\end{equation*}
Here, $f:[0,\infty)\times\br^d\times\PCal_D(\br^d)\times\br^d \to \br$ represents the running cost. $\{\mu_t\}_{t\geq0}$ is a flow of probability measures characterizing the mean-field information, with initial distribution $\mu_0=\mu^0\in\PCal^2(\br^d)$ and $\mu^0\perp\sigma(W_t,t\geq0)$. In the state dynamics, $\sigma>0$ is a constant diffusion coefficient; the drift term 
$b:[0,\infty)\times\br^d\times\PCal_D(\br^d)\times\br^d\to\br^d$
satisfies appropriate conditions ensuring the existence of a unique solution $\{X_t\}_{t\geq0}$ for the state dynamics such that for any $t\geq0$, $\mu_t=Law(X_t)\in\PCal^2(\br^d)$ (\cite{Rogers2000} and \cite{Evans1998}). Moreover, $\mu_t$ can be viewed as the limiting empirical distribution of identically distributed players' states, and by strong law of large numbers ${\mu_t=\lim_{N\to \infty} \frac{1}{N} \sum_{i=1}^N \delta_{X_t^{i}}} = Law(X_t)$ for all $t\in(0,T]$ where $X^i_t$ are i.i.d. copies of $X_t$. 

The Nash equilibrium solution  of this \eqref{eq:mfg-eg} is characterized  by the following optimality-consistency criterion.   
\begin{definition}\label{def:mfg-sol}
A control and mean-field pair $(\{\alpha^*_t\}_{t\geq0},\{\mu^*_t\}_{t\geq0})$, with initial distribution ${\mu^*_0=\mu^0}$, is called the solution to \eqref{eq:mfg-eg} if the following conditions hold. 
\begin{itemize}
 \item (Optimal control) Under $\{\mu^*_t\}_{t\geq0}$, $\{\alpha^*_t\}_{t\geq0}$ solves the following optimal control problem that: for $s\in[0,T]$ and $x\in\mathbb R^d$,
 \begin{equation*}
 \begin{aligned}
 u(s,x;\{\mu^*_t\})=&\inf_{\alpha \in \mathcal{A}} E\left[\int_s^T f(t,X_t,{\mu^*_t},\alpha_t)dt \left\vert X_s=x\right.\right] \\
 \mbox{subject to} \ \ \ \ \ \ \ \ \ & \\
 dX_t = &b(t,X_t,{\mu^*_t},\alpha_t)dt+\sigma dW_t,\qquad X_0\sim\mu^0.
 \end{aligned}
\end{equation*}
\item (Consistency) $\{\mu^*_t\}_{t\geq0}$ is the flow of probability distributions of the optimally controlled process, i.e., $\mu^*_t=Law(X^{*}_t)$ for $t\geq0$, where $X^{*}$ satisfies
\[dX^{*}_t=b(t,X^{*}_t,{\mu^*_t},\alpha^*_t)dt+\sigma dW_t,\qquad X^{*}_0\sim\mu^0.\]
\end{itemize}
\end{definition}
Correspondingly,  solutions of this  \eqref{eq:mfg-eg} can be characterized by the following coupled PDE system (assuming the mean-field interactions are of local type),
\begin{align}
&\left.\begin{aligned}
&\partial_su(s,x)+\frac{\sigma^2}{2}\Delta_xu(s,x)+H\left(s,x,\nabla_xu(s,x)\right)=0,\\
&\hspace{50pt} u(T,x)=0;\end{aligned}\right\}\label{eq:hjb}\tag{HJB}\\
&\left.\begin{aligned}
&\partial_sm(s,x)+\div\left[m(s,x)b(s,x,m(s,x),\alpha^*(s,x))\right]=\frac{\sigma^2}{2}\Delta_xm(s,x),\\
&\hspace{20pt} m(t,\cdot)\geq0,\,\int_{\br^d} m(t,x)dx=1,\,\forall t\in[0,T];\quad m(0,x)= m^0(x).\end{aligned}\right\}\label{eq:fp}\tag{FP}
\end{align}
Here, we denote by $m(t,\cdot)$ as the density function of $\mu_t$ for any $t\geq0$, and for any $\mu\in\PCal_D(\br^d)$ with density function $m$, we will write
${b(t,x,\mu,\alpha):=b(t,x,m,\alpha)}$ and ${f(t,x,\mu,\alpha):=f(t,x,m,\alpha)}$ in the rest of this paper. The Hamiltonian $H(s,x,p)$ in \eqref{eq:hjb} is given by
\[
H\left(s,x,p\right)=\min_{\alpha \in \br^d}\left\{b(s,x,m(s,x),\alpha)p+f(s,x,m(s,x),\alpha)\right\}, \qquad  s\in(0,T),\,x,p\in\mathbb R^d,\]
and $\alpha^*$ in \eqref{eq:fp} is the optimal control, with 
\begin{equation}
 \label{eq: opt-ctrl}
 \alpha^*(t,x)=\arg\min_{\alpha \in \br^d} \left\{b(t,x,m(t,x),\alpha)\nabla_xu(t,x)+f(t,x,m(t,x),\alpha)\right\}.
\end{equation}
Note that from \eqref{eq: opt-ctrl}, the optimal control is determined by the value function $u$ and mean field $m$.

Recently, a class of periodic MFGs
on the flat torus $\TT^d$ with a finite time horizon $[0,T]$ have been shown to admit a minimax structure by \cite{Cirant2018}. This minimax structure is   
further explored to connect MFGs with GANs in Section \ref{sec:gan2mfg} and \ref{subsec:mfg2gan}.

\subsection{Optimal transport problems}
The optimal transport problem, dating back to Monge \cite{monge1781memoire}, is to find the best transport plan minimizing the cost to transport one mass with distribution $\mu$ to another mass with distribution $\nu$.

Mathematically, consider a Polish space $\XCal$ with metric $l:\XCal\times\XCal\to[0,\infty)$. Let $\mu$ and $\nu$ be any probability measures in $\PCal(\XCal)$ with a first moment. Let $c:\XCal\times\XCal\to\br\bigcup\{
+\infty\}$ be a lower semi-continuous function such that $c(x,y)\geq a(x)+b(y)$, where $a$ and $b$ are upper semi-continuous functions from $\XCal$ to $\mathbb{R} \cup \{-\infty\}$ such that $a \in L^1(\XCal, \mu), b\in L^1(\XCal, \nu)$. The optimal transport problem is defined as follows (see \cite{villani2008optimal}). 
\begin{definition}
 The optimal transport cost between $\mu$ and $\nu$ with cost function $c$ is defined as 
\begin{equation}\label{eq:ot-div}\tag{OT}
 W_c(\mu,\nu)=\inf_{\pi\in\Pi(\mu,\nu)}\int_{\XCal\times\XCal}c(x,y)\pi(dx,dy),
\end{equation}
 where $\Pi(\mu,\nu)$ is the set of all possible couplings between $\mu$ and $\nu$.
\end{definition}

The well-definedness of this optimal transport problem, i.e., the existence of an optimal cost $W_c$, is guaranteed by Theorem 4.1 of \cite{villani2008optimal}. The corresponding dual formulation of this optimal transport problem is defined as follows. 
\begin{definition}
The dual Kantorovich problem of \eqref{eq:ot-div} is 
\begin{equation*}
\begin{aligned}
 D_c(\mu,\nu)=\sup_{\psi\in L^1(\XCal,\mu),\phi\in L^1(\XCal,\nu)}\biggl\{\int_\XCal\phi(x)\nu(dx)&-\int_\XCal\psi(x)\mu(dx)\biggl|\\
 &\phi(x)-\psi(y)\leq c(x,y),\,\forall (x,y)\in\XCal\times\XCal\biggl\}.
 \end{aligned}
\end{equation*}
\end{definition}

It is easy to see that $D_c(\mu,\nu)\leq W_c(\mu,\nu)$. The following Kantorovich-Rubinstein duality provides sufficient conditions to guarantee the equality. This duality is at the core of the relation between GANs and optimal transport to be established in the latter part of  this paper.
\begin{lemma}[Theorem 5.10(i) in \cite{villani2008optimal}]
\label{thm:kr-dual}
\[W_c(\mu,\nu)=D_c(\mu,\nu)=\sup_{\psi\in L^1(\XCal,\mu)}\int_{\XCal}\psi^c(x)\nu(dx)-\int_\XCal\psi(x)\mu(dx).\]
Here $\psi:\XCal\to\br\bigcup\{+\infty\}$ is taken from the set of all $c$-convex functions, i.e., $\psi$ is not constantly $+\infty$ and there exists a function $\zeta:\XCal\to\br\bigcup\{+\infty\}$ such that
\[\psi(x)=\sup_{y\in\XCal}\left[\zeta(y)-c(x,y)\right],\quad \forall x\in\XCal.\]
$\psi^c:\XCal\to\br\bigcup\{-\infty\}$ is its $c$-transform
\[\psi^c(y)=\inf_{x\in\XCal}\left[\psi(x)+c(x,y)\right],\quad \forall y\in\XCal.\]
\end{lemma}

In a particular case where the transport cost $c$ can be written as $c(x,y)=l(x,y)^p$ for some $p\geq1$, then the corresponding optimal cost gives rise to the Wasserstein distance between $\mu$ and $\nu$ of order $p$, or simply the Wasserstein-$p$ distance,
\[W_p(\mu,\nu)=\left[\inf_{\pi\in\Pi(\mu,\nu)}\int_{\XCal\times\XCal}l(x,y)^p\pi(dx,dy)\right]^{\frac{1}{p}}.\]
Note that for $p=1$, the Wasserstein-1 distance is adopted in Wasserstein GANs (WGANs) in \cite{Arjovsky2017} to improve the stability of GANs.




{
\section{GANs as MFGs}\label{sec:gan2mfg}


As reviewed in Section \ref{sec:prelim}, GANs have been introduced as two-player minimax games between the generator and the discriminator. In this section, we will take the viewpoint of the training data and re-examine the game theoretical perspective of GANs. In particular, we will show that under this viewpoint, the theoretical framework of GANs can be interpreted as MFGs under Pareto Optimality conditions:


\begin{theorem}\label{thm: gan2mfg}
GANs in \cite{Goodfellow2014} are MFGs under the Pareto Optimality criterion. 
\end{theorem}

To establish Theorem \ref{thm: gan2mfg}, we will start from a practical training scenario of GANs with a large amount of but finitely many data points $(Z_i,X_j)$, $i=1,\dots,N$, $j=1,\dots,M$ and $N,M$ finite. Here the $N$ latent variables $\{Z_i\}_{i=1}^N$ will be seen as $N$ players on the generator side and the $M$ true sample data $\{X_j\}_{j=1}^M$ will be seen as the adversaries on the discriminator side. We will show that after the training process, the optimal generator corresponds to the Pareto optimal strategy of a class of $N$-player cooperative games and it can only recover the sample distribution of $\{X_j\}_{j=1}^M$. Then, we will show that as $N,M\to\infty$, the $N$-player cooperative games will converge to the MFGs setting, corresponding to the theoretical GANs framework. 

\subsection{Cooperative game with N players against M adversaries}\label{subsec: gan2mfg-N-coop}
We will start from constructing the $N$-player cooperative games corresponding to a GANs model in training, with all players interchangeable. Let us first look at the state process of a representative generator player governed by the generator network $G$, parametrized by a feedforward neural network
\begin{equation}
    \label{eq: NN}
    NN=(T,\{\sigma_t\}_{t=1}^T,\{n_t\}_{t=0}^T).
\end{equation} Here, $T$ represents the depth of the neural network in which there are $T+1$ layers with layer 0 being to the input. $\sigma_t$ denotes the activation function adopted to produce the output at layer $t$, $t=1,\dots, T$. Finally, $n_t$ denotes the number of neurons at layer $t$, $t=1,\dots,T$. 

Recall that in GANs, the input of the generator $Z$ is a random variable from the latent space $\ZCal\subset\br^{n_0}$ and its distribution $\bp_z$ is known; the true samples follow an unknown distribution $\bp_r$. Let $X$ denote a random variable taking values in the sample space $\XCal\subset\br^{n_T}$ and $X\sim\bp_r$. 

\paragraph{State process of a representative player.}
Now, the state process of a player on the generator side is given by the feedforward process of generator network $G$. This computation process can be written as the following process,
\begin{equation}
    \label{eq: G-proc-generic}
    H_t^G=\sigma_t\left(w_tH_{t-1}^G+b_t\right),\quad H_0\sim\bp_z,\quad t=1,\dots T.
\end{equation}
Here, $H_0$ denotes the input  and $H_t$ denotes the output of layer $t$. $w_t$ takes values in $\br^{n_t\times n_{t-1}}$ and $b_t$ takes values in $\br^{n_t}$, $t=1,\dots,T$. Given $G$, denote $Law(H^G_T)=\bp_\theta$. 
Let $w_0=I_{n_0}$ be the identity matrix in $\br^{n_0\times n_0}$ and $b_0=0_{n_0}$ be the zero element in $\br^{n_0}$. Consider (deterministic) sequences $w=\{w_t\}_{t=0}^T$ and $b=\{b_t\}_{t=0}^T$. 
Notice that under a given feedforward neural network structure $(T,\{\sigma_t\}_{t=1}^T,\{n_t\}_{t=1}^T)$ the generator $G$ is characterized by $(w,b)=\{(w_t,b_t)\}_{t=0}^T$. In the analysis below, we will use $G$ and $(w,b)$ interchangeably whenever there is no risk of confusion. Denote the set of admissible controls by
\begin{equation}
    \label{eq: adm-ctrl}
    \begin{aligned}
        \mathcal G=\biggl\{G=(w,b)=\{(w_t,b_t)\}_{t=0}^T\biggl|&\,w_0=I_{n_0},\,b_0=0_{n_0},\\
        &\hspace{20pt}w_t\in\br^{n_t\times n_{t-1}},\,b_t\in\br^{n_t},\, t=1,\dots,T\biggl\}.
    \end{aligned}
\end{equation}

Given the mathematical setup for the state process above, we now move on to the definition of the cooperative game. Without loss of generality, we assume:
\begin{enumerate}
    \item $Z\perp X$. 
    \item Both $\bp_z$ and $\bp_r$ are absolutely continuous with respect to Lebesgue measure on $\XCal$, with density functions respectively $p_z$ and $p_r$. 
    \item The activation functions $\{\sigma_t\}_{t=1}^T$ are all continuous.
\end{enumerate}
\paragraph{The players and the adversaries.}
Under a strategy $G_i=(w_i,b_i)\in\mathcal G$, the state process for each individual player $i$ is given by
\begin{equation}
    \label{eq: G-proc-i}
    H_{i,t}^{G_i}=\sigma_t(w_{i,t}H_{i,t-1}^{G_i}+b_{i,t}),\quad H_{i,0}^{G_i}=Z_i\overset{i.i.d.}{\sim}\bp_z,\quad t=1,\dots,T.
\end{equation}
Let $\mathbf{G}^N=(G_1,\dots,G_N)$ denote a strategy profile for all $N$ players and $\bigotimes_{i=1}^N\mathcal G$ be the set of all possible strategy profiles. $\mathbf{G}^N\in\bigotimes_{i=1}^N\mathcal G$ is said to be {\em symmetric} if $G_i=G_j$ for any $i,j\in\{1,\dots,N\}$. Define the set of symmetric strategy profiles as
\begin{equation}
    \label{eq: n-sym-set}
    S(\mathcal{G})=\biggl\{\mathbf{G}^N\in\bigotimes_{i=1}^N\mathcal G\biggl|G=G_1=\dots=G_N,\,\forall G\in\mathcal G\biggl\}
\end{equation}
and for simplicity we write $\mathbf{G}^N=\mathbf{S}^{G}$ for any $\mathbf{G}^N=(G,\dots,G)\in S(\mathcal G)$.

The group of $N$ players are coordinated by a central controller who is in charge of choosing a strategy applied by all the players. In particular, since the players are indistinguishable, the central controller picks a strategy profile in $S(\mathcal G)$. Since $\mathbf S^G\in S(\mathcal G)$, the state process for each individual player is given by
\begin{equation}
    \label{eq: G-proc-ind}
    H_{i,t}^G=\sigma_t(w_tH_{i,t-1}^G+b_t),\quad H_{i,0}^G=Z_i\overset{i.i.d.}{\sim}\bp_z,\quad t=1,\dots,T,
\end{equation}
for $i=1,\dots,N$. Meanwhile, there are $M$ adversaries against this group of $N$ players, denoted by $X_j\in\XCal$, $j=1,\dots,M$, and $X_j\overset{i.i.d.}{\sim}\bp_r$.

\paragraph{The cost of the game.}
The cost for each individual player against the adversaries is measured by a discriminator function $D$ from the set of measurable functions \[\mathcal D=\{D \,|\, D:\XCal\to[0,1]\}.\]
Under any given $D\in\mathcal D$, the individual cost of choosing strategy $G\in\mathcal G$ is given by
\begin{equation}
    \label{eq: n-cost-ind}
    J_i^{N,M}\left(G;D,Z_i,\{X_j\}_{j=1}^M\right)=\log[1-D(H_{i,T}^G)]+\frac{1}{M}\sum_{j=1}^M\log D(X_j),\quad i=1,\dots, N,
\end{equation}
where we adopt the convention $0\log0=0$. 

The central controller from the player group will choose an optimal strategy profile $\mathbf{S}^G\in S(\mathcal G)$ to attain the minimal collective cost from all players even if the discriminator function $D^{N,M}$ is biased toward the group of adversaries in the following sense,
\begin{equation}
    \label{eq: n-D-adv}
    D^{N,M}=D^{N,M}\left(\cdot;G,\{Z_i\}_{i=1}^N,\{X_j\}_{j=1}^M\right)\in\argmax_{D\in\mathcal D}\frac{1}{N\cdot M}\sum_{i=1}^N\sum_{j=1}^M\left(\log D(X_j)+\log\left[1-D(H_{i,T}^G)\right]\right).
\end{equation}
That is, $D^{N,M}$ assigns higher scores to the adversaries compared with the players. The collective cost of choosing strategy profile $\mathbf{S}^G$ is given by
\begin{equation}
    \label{eq: n-cost}
    \begin{aligned}
    J^{N,M}\left(\mathbf{S}^G;\{Z_i\}_{i=1}^N,\{X_j\}_{j=1}^M\right)&=\frac{1}{N}\sum_{i=1}^NJ^{N,M}_i\left(G;D^{N,M};Z_i,\{X_j\}_{j=1}^M\right)\\
    &=\max_{D\in\mathcal D}\frac{1}{N\cdot M}\sum_{i=1}^N\sum_{j=1}^M\left(\log D(X_j)+\log\left[1-D(H_{i,T}^G)\right]\right).
    \end{aligned}
\end{equation}
Note that this is a deterministic quantity (given the values of $\{Z_i\}_{i=1}^N,\{X_j\}_{j=1}^M$) .

\paragraph{Pareto optimality of the cooperative game.}
The Pareto optimality to this cooperative game is specified as follows.
\begin{definition}[Pareto optimality]\label{defn: po}
A strategy profile $\mathbf S^{G^*}\in S(\mathcal G)$ is said to be Pareto optimal if 
\[J^{N,M}\left(\mathbf{S}^{G^*};\{Z_i\}_{i=1}^N,\{X_j\}_{j=1}^M\right)\leq J^{N,M}\left(\mathbf{S}^G;\{Z_i\}_{i=1}^N,\{X_j\}_{j=1}^M\right),\quad \forall\mathbf{S}^G\in S(\mathcal G).\]
\end{definition}
If such a Pareto optimal $\mathbf S^{G^*}$ exists for some $G^*\in\mathcal G$, then the corresponding game value is given by
\begin{equation}
    \label{eq: n-value}
    \begin{aligned}
        L^{N,M}\left(\{Z_i\}_{i=1}^N,\{X_j\}_{j=1}^M\right)&=\min_{\mathbf{S}^G\in S(\mathcal{G})}J^{N,M}\left(\mathbf{S}^G;\{Z_i\}_{i=1}^N,\{X_j\}_{j=1}^M\right)\\
        &=\min_{G\in\mathcal G}\max_{D\in\mathcal D}\frac{1}{N\cdot M}\sum_{i=1}^N\sum_{j=1}^M\log D(X_j)+\log\left[1-D(H_{i,T}^G)\right],
    \end{aligned}
\end{equation}
subject to \eqref{eq: G-proc-ind} under $G=G^*$ for all $i=1,\dots,N$.

Given the data points $\{Z_i\}_{i=1}^N$ and $\{X_j\}_{j=1}^M$ and a fixed strategy $G\in\mathcal{G}$, we first study the characterization through an optimality condition of $D^{N,M}$. Define
\begin{equation}
\label{eq:deltaMr-deltaNG}
    \mu_r^M=\frac{1}{M}\sum_{j=1}^M\delta_{X_j},\quad\mu_G^N=\frac{1}{N}\sum_{i=1}^N\delta_{H^G_{i,T}},
\end{equation}
as the empirical measures of $\{X_j\}_{j=1}^M$ and $\{H_{i,T}^G\}_{i=1}^N$, respectively, where for any given $y\in\XCal$, the Dirac delta function $\delta_y$ is given by
\[\delta_y(x)=\begin{cases}
1,&\text{if }x=y;\\
0,&\text{otherwise};
\end{cases}\quad \forall x\in\XCal,\]
and as Dirac measure,
\[\delta_y(A)=\begin{cases}
1,&\text{if }y\in A;\\
0,&\text{otherwise};
\end{cases}\]
for any Borel measurable set $A\subset\XCal$. Then
\begin{align}
    \frac{1}{N\cdot M}&\sum_{i=1}^N\sum_{j=1}^M\log D(X_j)+\log\left[1-D(H_{i,T}^G)\right]
    =&\int_{x\in\XCal} \log D(x)\mu_r^M(x)+\log\left[1-D(x)\right]\mu_G^N(x) dx\label{eq: n-D}.
\end{align}
Hence,
\begin{proposition}\label{prop: n-D-soln}
Any optimal biased discriminator function $D^{N,M}=D^{N,M}\left(\cdot;G,\{Z_i\}_{i=1}^N,\{X_j\}_{j=1}^M\right):\XCal\to[0,1]$ satisfies:
\begin{equation}
    \label{eq: n-D-adv-soln}
    D^{N,M}\left(x;G,\{Z_i\}_{i=1}^N,\{X_j\}_{j=1}^M\right)=
    \frac{\mu_r^M(x)}{\mu_r^M(x)+\mu_G^N(x)}, \qquad x\in\{H_{1,T}^G,\dots,H^G_{N,T},X_1,\dots,X_M\}.
\end{equation}
For $x\not\in\{H_{1,T}^G,\dots,H^G_{N,T},X_1,\dots,X_M\}$, $D^{N,M}\left(x;G,\{Z_i\}_{i=1}^N,\{X_j\}_{j=1}^M\right)$ can take any value in $[0,1]$.
\end{proposition}
Having the function $D^{N,M}$ in \eqref{eq: n-D-adv-soln}, we can now derive the Pareto optimal solution $\mathbf{S}^{G^*}$ in $S(\mathcal{G})$.
\begin{theorem}\label{thm: n}
Any Pareto optimal $\mathbf{S}^{G^{N,*}}$ by Definition \ref{defn: po} is given by some $G^{N,*}\in\mathcal{G}^{N,*}$ with
\begin{equation}
    \label{eq: n-G-po}
    \mathcal G^{N,*}=\biggl\{G\in\mathcal G: \mu_G^N=\mu_r^M\biggl\},
\end{equation}
where $\mu_G^N, \mu_r^M$ are defined in~\eqref{eq:deltaMr-deltaNG}, 
provided that $\mathcal G^{N,*}\neq\emptyset$.
\end{theorem}
Here, condition \eqref{eq: n-G-po} means that the empirical generator's output must match the empirical distribution of the samples. 

Theorem \ref{thm: n} shows that, in practical training of GANs over finite training samples,  the best an optimal generator can achieve is to recover the sample distribution of the true sample data. 
\begin{proof}
First, notice that 
\[\mu_r^M(x)\log D^{N,M}(x)+\mu_G^N(x)\log[1-D^{N,M}(x)]\leq0,\quad \forall x\in\XCal.\]
If there exists $x\in\{X_1,\dots,X_M\}\setminus\{H^G_{1,T},\dots,H^G_{N,T}\}$, then $D^{N,M}(x)=1$ and $\mu_r^M(x)\log D^{N,M}(x)+\mu_G^N(x)\log[1-D^{N,M}(x)]$ reaches its maximal value 0; similarly, if there exists $x\in\{H^G_{1,T},\dots,H^G_{N,T}\}\setminus\{X_1,\dots,X_M\}$, then $D^{N,M}(x)=0$ and again $\mu_r^M(x)\log D^{N,M}(x)+\mu_G^N(x)\log[1-D^{N,M}(x)]$ reaches its maximal value 0. Therefore, the supports of $\mu_r^M$ and $\mu_G^N$ should coincide. Furthermore,
\[\begin{aligned}
    &J^{N,M}\left(\mathbf S^G;\{Z_i\}_{i=1}^N,\{X_j\}_{j=1}^M\right)\\
    &\hspace{30pt}=\sum_{x\in\{X_1,\dots,X_M\}}\mu_r^M(x)\log\frac{\mu_r^M(x)}{\mu_r^M(x)+\mu_G^N(x)}+\mu_G^N(x)\log\frac{\mu_G^N(x)}{\mu_r^M(x)+\mu_G^N(x)}\\
    &\hspace{30pt}=-\log4+\sum_{x\in\{X_1,\dots,X_M\}}\mu_r^M(x)\log\mu_r^M(x)\\
    &\hspace{60pt}\sum_{x\in\{X_1,\dots,X_M\}}\mu_G^N(x)\log\mu_G^N(x)-[\mu_r^M(x)+\mu_G^N(x)]\log\frac{\mu_r^M(x)+\mu_G^N(x)}{2}\\
    &\hspace{30pt}=:F\left(\mu_G^N(x);\{Z_i\}_{i=1}^N,\{X_j\}_{j=1}^M\right).
\end{aligned}\]
Let $\PCal\left(\{X_1,\dots,X_M\}\right)$ denote the set of probability measures on $\{X_1,\dots,X_M\}$. Then the functional $F\left(\cdot;\{Z_i\}_{i=1}^N,\{X_j\}_{j=1}^M\right):\PCal\left(\{X_1,\dots,X_M\}\right)\to\br$ is uniquely minimized at $\mu_G^N=\mu_r^M$. The conclusion follows.
\end{proof}

\subsection{Mean-field game when N goes to infinity}\label{subsec: gan2mfg-MF}
Now as $N$ and $M$ tend to infinity, instead of assigning to  each and every one of the player a strategy, the central controller now considers 
\[\mu^G=\left\{\mu^G_t=\lim_{N\to\infty}\frac{\sum_{i=1}^N\delta_{H^G_{i,t}}}{N}\right\}_{t=0}^T.\]
We refer $\mu^G$ as the mean-field information. Given the i.i.d. conditions on $H_{i,0}^G$ and the symmetric strategy profile, we can write $\mu^G_t=Law(H^G_t)$ for $t=0,\dots, T$, with $H^G$ given by \eqref{eq: G-proc-generic}. Note that $H^G_{i,T}\sim\mu^G_T=\bp_\theta$ for all $i=1,2,\dots$. 

Now, instead of playing against $M$ adversaries from the discriminator side as in the training scenario, the group of infinitely many generator players are now competing against the aggregated information of the adversaries, that is, the distribution of the adversaries $\bp_r$. Under a given discriminator function $D\in\mathcal D$, the cost of individual players for choosing strategy $G\in\mathcal G$ is given by
\begin{equation}
    \label{eq: mf-cost-ind}
    J_i(G;D,Z_i,\bp_r)=\log\left[1-D(H_{i,T}^G)\right]+\EE_{X\sim\bp_r}\left[\log D(X)\right],
\end{equation}
for each $i=1,2,\dots$. The objective of the central controller is to choose an optimal symmetric strategy profile characterized by $G\in\mathcal{G}$ such that the collective cost of all players minimized even under a biased discriminator function $\bar D$ toward the adversary group, such that
\begin{equation}
    \label{eq: mf-D-adv}
    \begin{aligned}
    \bar D=\bar D(\cdot;G,\bp_z,\bp_r)\in&\argmax_{D\in\mathcal{D}}\EE_{X\sim\bp_r}[D(X)]+\lim_{N\to\infty}\frac{1}{N}\sum_{i=1}^N\log\left[1-D(H^G_{i,T})\right]\\
    \overset{a.s.}{=}&\argmax_{D\in\mathcal{D}}\EE_{X\sim\bp_r}[\log D(X)]+\EE_{Y\sim\mu^G_T}[\log[1-D(Y)]],
    \end{aligned}
\end{equation}
where the equality is by the law of large numbers. Following a similar argument as for Proposition \ref{prop: n-D-soln}, we recover the following characterization of $\bar D$ from~\cite{Arjovsky2017}. 
\begin{proposition}
\label{prop: mf-D-soln}
Any optimal biased discriminator $\bar D$ satisfies: 
\begin{equation}
    \label{eq: mf-D-soln}
    \bar D(x;G,\bp_z,\bp_r)=
    \frac{p_r(x)}{p_r(x)+p_\theta(x)},
    \qquad x\in{\rm Supp}(p_r)\bigcup{\rm Supp}(p_\theta),
\end{equation}
where ${\rm Supp}(p_r)$ and ${\rm Supp}(p_\theta)$ denote the supports of $p_r$ and $p_\theta$, respectively. For $x\not\in{\rm Supp}(p_r)\bigcup{\rm Supp}(p_\theta)$, $\bar D$ can take any value in [0,1].
\end{proposition}

Under this biased discriminator function $\bar D$, the collective cost of all players for choosing strategy $G$ is then given by
\begin{equation}
    \label{eq: mf-cost}
    \begin{aligned}
    J(G;\bp_z,\bp_r)&=\lim_{N\to\infty}\frac{1}{N}J_i(G;\bar D,\bp_z,\bp_r)\\
    &\overset{a.s.}{=}\max_{D\in\mathcal{D}}\EE_{X\sim\bp_r}[\log D(X)]+\EE_{Y\sim\mu^G_T}[\log[1-D(Y)]],
    \end{aligned}
\end{equation}
where the last equality is also due to the law of large numbers. The game value for the mean-field game under Pareto optimality condition is then given by
\begin{equation}
    \label{eq: mf-value}
    \begin{aligned}
    L(\bp_z,\bp_r)&=\min_{G\in\mathcal{G}}J(G;\bp_z,\bp_r)\\
    &=\min_{G\in\mathcal G}\max_{D\in\mathcal{D}}\EE_{X\sim\bp_r}[\log D(X)]+\EE_{Y\sim\mu^G_T}[\log[1-D(Y)]],
    \end{aligned}
\end{equation}
subject to the state process \eqref{eq: G-proc-generic} and $\mu^G=\{\mu^G_t\}_{t=0}^T$ given by $\mu^G_t=Law(H_t^G)$ for $t=0,\dots,T$. With a similar proof as for Theorem \ref{thm: n}, we have the following result for solving the mean-field game \eqref{eq: mf-value}.
\begin{theorem}
\label{thm: mf}
The mean-field game \eqref{eq: mf-value} can be solved by any $G^{mf, *}\in\mathcal{G}^{mf,*}$ where
\begin{equation}
    \label{eq: mf-G-po}
    \mathcal G^{mf,*}=\biggl\{G\in\mathcal G: \mu^G_T=\bp_r\biggl\},
\end{equation}
provided that $\mathcal G^{mf, *}\neq\emptyset$.
\end{theorem}
By the continuous mapping theorem, we have the following connection between the $N$-player game \eqref{eq: n-value} and the mean-field games \eqref{eq: mf-value}.
\begin{theorem}
\label{thm:convergence-D-G}
Under any given $G\in\mathcal G$ and for any $x\in\XCal$, as $N\to\infty$ and $M\to\infty$,
\[D^{N,M}(x)\overset{a.s.}{\longrightarrow}\bar D(x),\quad \mathcal{G}^{N,*}\overset{a.s.}{\longrightarrow}\mathcal{G}^{mf, *}.\]
\end{theorem}
Notice that under a neural network structure specified by \eqref{eq: NN}, the weights $w^*=\{w^*_t\}_{t=0}^T$ and biases $b^*=\{b^*_t\}_{t=0}^T$ given by $G^{mf,*}=(w^*,b^*)\in\mathcal{G}^{mf,*}$ coincides with the training objective of the vanilla GANs structure proposed in \cite{Goodfellow2014}. Therefore, Theorem \ref{thm: gan2mfg} follows directly from Theorem \ref{thm:convergence-D-G}. 

\begin{remark}
\begin{enumerate}
    \item The $N$-player cooperative game corresponds to the training of GANs in practice, where the training is over a finite dataset $\left\{\{Z_i\}_{i=1}^N,\{X_j\}_{j=1}^M\right\}$, whereas the mean-field game corresponds to the theoretical frame of GANs where both $\bp_z$ and $\bp_r$ can be perfectly simulated on an infinite supply of data. The non-emptiness of $\mathcal{G}^{N,*}$ and $\mathcal{G}^{mf, *}$ will rely on the design of the generator network architecture.
    \item Mean-field games under Pareto Optimality condition are often referred to as the mean-field control problems, see for instance \cite{Bensoussan2013,Carmonaa}.
\end{enumerate}
\end{remark}}

\section{GANs and optimal transport}\label{sec:gan-ot}

In this section, we will show that, from the perspective of a dynamical version of optimal transport, GANs are optimal transport problems from the known prior distribution $\bp_z$ of the latent variable $Z\in\ZCal$ to the unknown true distribution $\bp_r$ over the sample space $\XCal$. 

Intuitively, this connection between the optimal transport problem and GANs is as follows. In the optimal transport problem  \cite{benamou2000computational},  the mass transport from a source distribution to a target distribution is seen as the evolution of a density field from the source density function to the target density function according to a given differential equation characterized by a carefully chosen velocity field. In the context of GANs, the evolution of the density flow is cast as the forward pass within the generator network \(G\), as specified in \eqref{eq: G-proc-generic}. 
Let \(\{\mu_t^G=Law(H_t^G)\}_{t=0}^T\) denote the flow of probability measures generated by this dynamic, with initial distribution \(\mu_0^G=\bp_z\). Then the goal of GANs is to find an optimal generator network \(G^*\), which is also an optimal transport map, such that both the initial condition \(\mu_0^{G^*}=\bp_z\) and the terminal condition \(\mu_T^{G^*}=\bp_r\) hold. Note that such $G^*$'s may not be unique, as pointed out in \cite{benamou2000computational}. Nevertheless, with a proper choice of transport cost, i.e., the loss function in GANs, one particular transport map $G^*$ can be determined. In fact, this optimal transport viewpoint of GANs can be explicitly constructed  in the case of WGANs.

To establish this connection rigorously, 
recall that the latent variable $Z\in\ZCal$ with $Z\sim\bp_z$ and the sample space $\XCal$ is a metric space with distance function $l:\XCal\times\XCal\to\br^+$. Fix an arbitrary $x_0\in\XCal$ and denote
\begin{equation}
    \label{eq: adm-G}
    \mathcal{G}(l,x_0)=\biggl\{G:\ZCal\to\XCal:G\text{ continuous and }Law(G(Z))\in L^1(\XCal;x_0)\biggl\}.
\end{equation}
Assume $\bp_r\in L^1(\XCal;x_0)$ and denote by $\Pi(\bp_z,\bp_r)$ the set of all possible couplings of $\bp_z$ and $\bp_r$. For any fixed $G\in\mathcal{G}(l,x_0)$, define the following transport cost function $c_G:\ZCal\times\XCal\to\RR^+$ such that
\begin{equation}
    \label{eq: c-G}
    c_G(z,x)=l\left(G(z),x\right),\quad \forall z\in\ZCal,\,x\in\XCal.
\end{equation}
Accordingly, we have the following optimal transport problem from $\bp_z$ to $\bp_r$:
\begin{equation}
    \label{eq: ot-g}\tag{OT-G}
    \WCal_{c_G}(\bp_z,\bp_r)=\inf_{\pi\in\Pi(\bp_z,\bp_r)}\int_{\ZCal\times\XCal}c_G(z,x)\pi(dz,dx).
\end{equation}
Its corresponding dual problem is then given by
\begin{equation}
    \label{eq: dual-G}
    \begin{aligned}
    &D_{c_G}(\bp_z,\bp_r)=\\
    &\sup_{\psi\in L^1(\ZCal,\bp_z),\phi\in L^1(\XCal,\bp_r)}\biggl\{\int_{\XCal}\phi(x)\bp_r(dx)-\int_{\ZCal}\psi(z)\bp_z(dz):\phi(x)-\psi(z)\leq c_G(z,x),\,\forall(z,x)\in\ZCal\times\XCal\biggl\}.
    \end{aligned}
\end{equation}
Analogously to \cite[Definitions 5.2 \& 5.7]{villani2008optimal}, we have the following definition tailored to $c_G$.
\begin{definition}[$c_G$-convexity and $c_G$-concavity]\label{def: c-g-conv-conc}
A function $\psi:\ZCal\to\br\cup\{+\infty\}$ is $c_G$-convex if it is not constantly $+\infty$ and there exists a function $\zeta:\XCal\to\br\cup\{+\infty\}$ such that 
\[\psi(z)=\sup_{x\in\XCal}\biggl[\zeta(x)-c_G(z,x)\biggl],\quad\forall z\in\ZCal;\]
its corresponding $c_G$-transform is 
\[\psi^c(x)=\inf_{z\in\ZCal}\biggl[\psi(z)+c_G(z,x)\biggl],\quad \forall x\in\XCal.\]
A function $\phi:\XCal\to\br\cup\{+\infty\}$ is $c_G$-concave if it is not constantly $+\infty$ and there exists a $c_G$-convex function $\psi$ such that $\phi=\psi^c$; its $c_G$-transform is given by
\[\phi^c(z)=\sup_{x\in\XCal}\biggl[\phi(x)-c_G(z,x)\biggl].\]
\end{definition}
Now, denote
\begin{equation}
    \label{eq: 1-lip}
    Lip^1(\XCal,l)=\biggl\{f:\XCal\to\br:f(x)-f(y)\leq l(y,x),\,\forall(x,y)\in\XCal\times\XCal\biggl\}.
\end{equation}

Note that $Lip^1(\XCal,l)$ is equivalent to the set of 1-Lipschitz functions from $\XCal$ to $\br$ with respect to metric $l$. Moreover, we have
\begin{lemma}
\label{lem: lip-conc}
Suppose that $G(\ZCal):=\{G(z):\forall z\in\ZCal\}=\XCal$. Then $Lip^1(\XCal,l)$ coincides with the set of $c_G$-concave functions.
\end{lemma}
\begin{proof}
We first prove that $\phi\in Lip^1(\XCal,l)$ is included in the set of $c_G$-concave functions. Let $\phi\in Lip^1(\XCal,l)$. Then for any $(z,x)\in\ZCal\times\XCal$, $\phi(x)<+\infty$ and
\[\begin{aligned}
\phi(x)-c_G(z,x)&=\phi(G(z))+\phi(x)-\phi(G(z))-c_G(z,x)\\
&=\phi(G(z))+\phi(x)-\phi(G(z))-l(G(z),x)\\
&\leq\phi(G(z))\\
&=\phi(G(z))-c_G(z,G(z)).
\end{aligned}\]
Note that there is equality if $x = G(z)$. Therefore, $\phi\circ G:\ZCal\to\br$ is $c_G$-convex. Since $G(\ZCal)=\XCal$, for any $x\in\XCal$, there exists $z(x)\in\ZCal$ such that $G(z(x))=x$. For any $(z,x)\in\ZCal\times\XCal$,
\[\begin{aligned}
\phi(G(z))+c_G(z,x)&=\phi(G(z))+l(G(z),x)\\
&=\phi(x)-\left[\phi(x)-\phi(G(z))-l(G(z),x)\right]\\
&\geq \phi(x)=\phi(G(z(x))).
\end{aligned}\]
Therefore, $\phi=\biggl(\phi\circ G\biggl)^c$ and hence $\phi$ is $c_G$-concave.

We now prove the converse inclusion. For the sake of contradiction, suppose there exists a $c_G$-concave function $\phi_0$ that is not 1-Lipschitz. Due to $c_G$-concavity, there exists a $c_G$-convex function $\psi_0:\ZCal\to\br\cup\{+\infty\}$ such that $\phi_0=\psi_0^c$. Since $\psi_0$ is not constantly $+\infty$, given the form of $\psi_0^c$ and $c_G$, $\phi_0$ is finite everywhere. Hence, $\phi_0$ not being 1-Lipschitz implies that there exists $x,y\in\XCal$ such that $\phi_0(x)-\phi_0(y)>l(y,x)$. We have 
\[\begin{aligned}
&\phi_0(x)=\inf_{z\in\ZCal}\biggl[\psi_0(z)+c_G(z,x)\biggl],\\
&\phi_0(y)=\inf_{z\in\ZCal}\biggl[\psi_0(z)+c_G(z,y)\biggl].
\end{aligned}\]
For any $\epsilon>0$, there exists $z(\epsilon)\in\ZCal$ such that 
\[\phi_0(y)>\psi_0(z(\epsilon))+c_G(z(\epsilon),y)-\epsilon.\]
Then
\begin{equation}\label{eq: contrad-e}
\begin{aligned}
\phi_0(x)-\phi_0(y)&<\phi_0(x)-\psi_0(z(\epsilon))-c_G(z(\epsilon),y)+\epsilon\\
&\leq \psi_0(z(\epsilon))+c_G(z(\epsilon),x)-\psi_0(z(\epsilon))-c_G(z(\epsilon),y)+\epsilon\\
&=l(G(z(\epsilon)),x)-l(G(z(\epsilon)),y)+\epsilon\\
&\leq l(y,x)+\epsilon,
\end{aligned}
\end{equation}
where the equality is by definition of $c_G$. 
Since \eqref{eq: contrad-e} holds for any $\epsilon>0$, we arrive at a contradiction that $\phi_0(x)-\phi_0(y)\leq l(x,y)$.

Therefore, a function $\phi:\XCal\to\br\cup\{+\infty\}$ is $c_G$-concave if and only if it is 1-Lipschitz with respect to metric $l$.
\end{proof}
 Note that an informal form of the lemma can be found in \cite{villani2008optimal} regarding the equivalence between the set of 1-Lipschitz functions and the set of $c$-concave functions  when computing Wasserstein-1 distance between two distributions in $\mathcal{P}(\XCal)$. Here we establish rigorously this equivalence for optimal transport problem from  $\bp_z\in\mathcal{P}(\ZCal)$ to $\bp_r\in\mathcal{P}(\XCal)$.



This lemma,  the foundational piece of a duality theory for the optimal transport problem \eqref{eq: ot-g}, enables us to connect GANs and optimal transport from the latent distribution $\bp_z$ to the true distribution $\bp_r$. 
The duality is explicit in the case of WGANs, and can be utilized to cover a broad class of GANs including the relaxed Wasserstein GANs \cite{guo2017relaxed}, in which  the metric function is replaced with the Bregman divergence.

\begin{theorem}\label{thm: wgan-ot}
Suppose that $\PP_r\in L^1(\XCal;x_0)$ for some $x_0\in\XCal$.
WGAN is equivalent to the following optimization problem,
\[\min_{G\in\{G\in\mathcal G(l,x_0):G(\ZCal)=\XCal\}}\WCal_{c_G(\bp_z,\bp_r)},\]
where $\WCal_{c_G}(\bp_z,\bp_r)$ is given by \eqref{eq: ot-g}.
\end{theorem}
\begin{proof}
For any $G\in\{G\in\mathcal G(l,x_0):G(\ZCal)=\XCal\}$, $c_G(z,x)\geq l(G(z),x_0)-l(x,x_0)=a(z)+b(x)$ for any $(z,x)\in\ZCal\times\XCal$. Since $a(\cdot)=l(G(\cdot),x_0)\in L^1(\ZCal,\bp_z)$ and $b(\cdot)=l(\cdot,x_0)\in L^1(\XCal,\bp_r)$, by \cite[Theorem 5.10]{villani2008optimal} and Lemma \ref{lem: lip-conc},
\[\begin{aligned}
    \WCal_{c_G}(\bp_z,\bp_r)&=D_{c_G}(\bp_z,\bp_r)\\
    &=\max_{D\in Lip^1(\XCal,l)}\mathbb{E}_{X\sim\bp_r}D(X)-\mathbb{E}_{Z\sim\bp_z}D\circ G(Z).
\end{aligned}\]
Notice the form of $c_G$ in \eqref{eq: c-G} that $c_G(z,x)=l(G(z),x)$ for any $(z,x)\in\ZCal\times\XCal$, we have 
\[\argmin_{(x,y)\in\XCal\times\XCal}l(x,y)=\{(x,y):x=y\}.\]
If there exists $\pi^*\in\Pi(\bp_z,\bp_r)$ such that $1 = \pi^*(\ZCal\times\XCal)=\pi^*(\{(z,x):G^*(z)=x\})$ for some $G^*\in\mathcal{G}(l,x_0)$ with $G^*(\ZCal)=\XCal$, then
\[\int_{\ZCal\times\XCal}c_{G^*}(z,x)\pi^*(dz,dx)=\min_{G\in\{G\in\mathcal G(l,x_0):G(\ZCal)=\XCal\}}\WCal_{c_G}(\bp_z,\bp_r).\]
Furthermore, for any $B_x\subset\XCal$, 
\[\bp_r(B_x)=\int_{\ZCal\times B_x}\pi^*(dz,dx)=\int_{(G^*)^{-1}(B_x)\times B_x}\pi^*(dz,dx)=\int_{(G^*)^{-1}(B_x)\times\XCal}\pi^*(dz,dx)=G^*\#\bp_z(B_x),\]
i.e. $Law(G(Z))=\bp_r$ where $Z\sim\bp_z$.
\end{proof}
To extend this explicit connection between WGANs and the  optimal transport to other variations such as relaxed Wasserstein GANs,   consider a quasi-metric $\bar l:\XCal\times\XCal\to\br$ instead of a metric function $l$ such that
\begin{enumerate}
    \item $\bar l(x,y)\geq0$ and the equality holds if and only if $x=y$;
    \item $\bar l(x,y)\leq \bar l(x,z)+\bar l(z,y)$ for any $x,y,z\in\XCal$.
\end{enumerate}
That is, $\bar l$ relaxes the symmetric property of $l$. Then, for a fixed $x_0\in\XCal$, define 
\[\mathcal{G}(\bar l,x_0)=\biggl\{G:\ZCal\to\XCal:\,G\text{ continuous and }\int_{\ZCal}\bar l(G(z),x)\bp_z(dz)<\infty\biggl\}.\]
For any $G\in\mathcal{G}(\bar l,x_0)$, let $\bar c_{G}:\ZCal\times\XCal\to\br$ such that $\bar c_G(z,x)=\bar l(G(Z),x)$ for any $(x,z)\in \ZCal\times\XCal$.
\begin{theorem}
\label{thm: wgan-ot-2}
Suppose that $\bp_r$ satisfies $\int_\XCal\bar l(x_0,x)\bp_r(dx)<\infty$.
\begin{enumerate}
    \item For any $G\in\mathcal{G}(\bar l,x_0)$ such that $G(\ZCal)=\XCal$, a function $\phi$ is $\bar c_G$-concave if and only if 
    \[\phi\in Lip^1(\XCal,\bar l)=\biggl\{f:\XCal\to\br:\, f(x)-f(y)\leq \bar l(y,x),\,\forall (x,y)\in\XCal\times\XCal\biggl\}.\]
    \item The following equivalent condition holds
    \begin{equation*}
        \min_{G\in\{G\in\mathcal{G}(\bar l, x_0): G(\ZCal)=\XCal\}}\max_{D\in Lip^1(\XCal,\bar l)}\mathbb{E}_{\bp_r}D(X)-\mathbb{E}_{\bp_z}D\circ G(Z)\Leftrightarrow\min_{G\in\{G\in\mathcal{G}(\bar l, x_0): G(\ZCal)=\XCal\}}\WCal_{\bar c_G}(\bp_z,\bp_r).
    \end{equation*}
\end{enumerate}
\end{theorem}
The proof of Theorem \ref{thm: wgan-ot-2} follows the same arguments in the proofs of Lemma \ref{lem: lip-conc} and Theorem \ref{thm: wgan-ot}.

\section{Computing MFGs via GANs}\label{sec:num}
In Section \ref{sec:gan2mfg},  GANs are conceptually  interpreted as MFGs among players whose state process is captured by the forward pass within the generator network. In this section, we will show the inverse interpretation of  MFGs  as GANs, from which a new computational approach for solving MFGs is proposed and tested through several examples.

\subsection{MFGs as GANs}\label{subsec:mfg2gan}
\paragraph{MFGs recast as GANs.}
 First, we claim that
 {\em MFGs can be recast as GANs.}

To see this, recall from Section \ref{sec:prelim} that in classical GANs, the generator $G$ mimics the sample data to generate new ones in order to minimize the difference between the true distribution $\bp_r$ and $\PP_{\theta}$. The discriminator $D$ measures the performance of the generator by some divergence between $\bp_\theta$ and $\bp_r$. Meanwhile, MFGs can be recast in this context such that
\begin{itemize}
 \item The latent space $\ZCal=\br^d$ and sample $x$ of latent variable $Z$ are drawn from the probability distribution $\PP_z=\mu^0$.
 \item The value function $u$ maps the element $x$ into $\br$ so that it gives the optimal cost and its gradient dictates the optimal strategy in the equilibrium state of the MFGs. We can then define a loss function 
 \[L_{Val}(u,m)=L_{HJB}(u,m)+\beta_{Val}L_{term}(u,m),\]
 where $\beta_{Val}>0$ denotes the weight on the penalty of the terminal condition, with
\[L_{HJB}(u,m)=\int_0^T\int_{\br^d}\left|\partial_t u(s,x)+F\left(s,x,m(s,x),\nabla_xu(s,x), \Delta_xu(s,x)\right)\right|^2\mu(dx,ds),\]
 and
\[
    L_{term}(u,m)=\int_{\br^d}|u(T,x)-g(x, m(T,x))|^2\mu^{term}(dx).
\]
 \item The equilibrium state of the MFGs,  parallel to the true distribution $\bp_r$ in GANs, is characterized through a {\em consistency condition} between the value function and the controlled dynamics.
 The mean-field term $m$ attempts to approximate the equilibrium state process \eqref{eq: state-evol} under the optimal control given by the value function. Its loss, in place of the divergence function between $\PP_{\theta}$ and ${\PP_r}$,
 is defined as 
 \[L_{MF}(u,m)=L_{FP}(u,m)+\beta_{MF}L_{init}(u,m),\]
 where $\beta_{MF}>0$ denotes the weight on the penalty of the initial condition, with
  \[\begin{aligned}
 L_{FP}(u,m)&=\int_0^T\int_{\br^d}\biggl|\partial_t m(s,x)+\div\left[m(s,x)b(s,x,m(s,x),\alpha^*_{u,m}(s,x))\right]\\
 &\hspace{120pt}-\frac{\sigma^2}{2}\Delta_xm(s,x)\biggl|^2\mu(dx,ds),
 \end{aligned}
 \]
 and
 \[
    L_{init}(m)=\int_{\br^d}\biggl|m(0,x)-m^0(x)\biggl|^2\mu^{init}(dx).
\]
 Here $\alpha^*_{u,m}$ denotes the optimal control solved under current $u$ and $m$. 
\end{itemize}

Now MFGs can be recast as GANs with $m$ as a generator and $u$ as a discriminator. In this case, $m$ tries to learn how to mimic the equilibrium mean-field distribution and $u$ learns how to penalize samples that are deviating from this equilibrium, as summarized in  Table~\ref{tab:gan-mfg-1}.

\begin{table*}[!ht]
\centering
\caption{Link between GANS and MFGs}
\begin{tabular*}{1.0\textwidth}{>{\bfseries}p{0.22\textwidth} p{0.28\textwidth} p{0.40\textwidth}}
\toprule
& {\bf GANs} & {\bf MFGs} \\
\toprule
Generator G & neural network for approximating the map $G:\mathcal Z\mapsto \XCal$ & neural network for mean-field $m$ (solving FP) \\
\midrule
Discriminator D & neural network measuring divergence between $\mathbb{P}_\theta$ and $\mathbb{P}_r$ & neural network for value function $u$ (solving HJB)\\
\bottomrule
\end{tabular*}
\label{tab:gan-mfg-1}
\end{table*}

For some classes of MFGs, this connection with GANs is even more explicit, as we explain next.

\paragraph{Example.} 
 Let us review the following class of periodic MFGs on the flat torus $\TT^d$ from \cite{Cirant2018} in which the individual agent's cost is:
\begin{equation}\label{eq:dyn-cost}
 J_m(\alpha) = \EE\left[\int_0^T \left(L(X^\alpha_t, \alpha(t,X^\alpha_t)) + f(X^\alpha_t, m(t,X^\alpha_t)) \right)dt + u^T(X^\alpha_T)\right],
\end{equation}
where $X^\alpha = (X^\alpha_t)_t$ is a $d$-dimensional process whose initial distribution has density $m^0$ such that
\[
 dX^\alpha_t = \alpha(t,X^\alpha_t) dt + \sqrt{2 \epsilon} d W_t.
\]
Here $\alpha: [0,T] \times \RR^d \to \RR^d$ is a control policy, $L$ and $f$ constitute the running cost,  and $m(t,\cdot)$, for $t\in[0,T]$, denotes a probability density which corresponds to the law of $X^\alpha_t$ at equilibrium. 

Consider the convex conjugate of the running cost $L$, defined as:
\begin{equation}
    \label{eq:defH0}
 H_0(x, p) = \sup_{\alpha \in \RR^d} \left\{ \alpha \cdot p - L(x, \alpha) \right\},
\end{equation}
and let $F(x,m)=\int^mf(x,z)dz$ for $x \in \RR^d$, $m \in [0,+\infty)$. 

This class of MFGs can be characterized by the following coupled PDE system corresponding to ~\eqref{eq:hjb}--\eqref{eq:fp},
\begin{equation}
 \label{eq:dyn-hjb-fp}
 \begin{cases}
 &-\partial_s u-\epsilon\Delta_x u+H_0(x,\nabla_x u)=f(x,m),\\
 &\partial_sm-\epsilon\Delta_x m-\div\left(m\nabla_pH_0(x,\nabla u)\right)=0, \\
 &m>0,\,m(0,\cdot)=m^0(\cdot),\,u(T,\cdot)=u^T(\cdot),
 \end{cases}
\end{equation}
where the first equation is an HJB equation governing the value function and the second is a FP equation governing the evolution of the optimally controlled state process. 
The system of equations \eqref{eq:dyn-hjb-fp} can be shown to be the system of optimality conditions for the following minimax game:
\begin{equation}
 \label{eq:dyn-var-str}
 \inf_{u\in\mathcal C^2([0,T]\times\mathbb T^d)}\sup_{m\in\mathcal C^2([0,T]\times\mathbb T^d)}\Phi(m,u),
\end{equation}
where
\begin{equation*}
 \label{eq:dyn-functional}
 \begin{aligned}
 \Phi(m,u)&=\int_0^T\int_{\mathbb T^d}\left[m(-\partial_t u-\epsilon\Delta_x u)+mH_0(x,\nabla_x u)-F(x,m)\right] dxdt\\
 &\hspace{10pt}+\int_{\mathbb T^d}\left[m(T,x)u(T,x)-m^0(x)u(0,x)-m(x,T)u^T(x)\right]dx.
 \end{aligned}
\end{equation*}
  The connection between GANs and MFGs is transparent from \eqref{eq:dyn-var-str}.

\subsection{Computing MFGs via GANs }
The interpretation of MFGs as GANs points to a new computational approach for MFGs using two competing neural networks, assuming that the equilibrium of MFGs can be computed via the coupled HJB-FB system. 
That is, one can compute MFGs using two neural networks trained in an adversarial way as in GANs, with:
\begin{itemize}
 \item $u_\theta$ being the neural network approximation of the unknown value function $u$ for the HJB equation, 
 \item $m_\omega$ being the neural network approximation for the unknown mean field information function $m$.
\end{itemize}
This new computational algorithm for MFGs is summarized in Algorithm \ref{alg:mfgan-dyn}. Note that here  the two neural networks are trained in alternating periods of length respectively $N_\omega$ and $N_\theta$. 
Furthermore, each neural network is trained to minimize a loss function hence the iterations are a gradient descent in each case. 

In the algorithm, for $B_g$ samples $(\underline{s},\underline{x})=\{(s_i,x_i)\}_{i=1}^{B_g}$ to train the generator, the empirical loss is defined as:
\begin{align}
\label{eq:algo-defLmf}
    \hat L_{MF}(\theta,\omega; \underline{s},\underline{x})=\hat L_{FP}(\theta,\omega; \underline{s},\underline{x})+\beta_{MF} \hat L_{init}(\omega; \underline{s},\underline{x}),
\end{align}
  with the FP PDE residual's empirical loss and the initial condition's empirical loss being defined as: 
  $$
 \begin{aligned}
 \hat L_{FP}&=\frac{1}{B_d}\sum_{i=1}^{B_d}\biggl|\partial_sm_\omega(s_i,x_i)+\div\left[m_\omega(s_i,x_i)b(s_i,x_i,m(s_i,x_i),\alpha^*_{\theta,\omega}(s_i,x_i))\right]-\frac{\sigma^2}{2}\Delta_xm_\omega(s_i,x_i)\biggl|^2,\\
 \hat L_{init}&=\frac{1}{B_d}\sum_{i=1}^{B_d}\left[m_\omega(0,x_i)-m^0(x_i)\right]^2,
 \end{aligned}
 $$
 where $m^0$ is a known density function for the initial distribution of the states and $\beta_{MF}>0$ is the weight for the penalty on the initial condition of $m$.
 Moreover, for $B_d$ samples  $(\underline{s},\underline{x})=\{(s_i,x_i)\}_{i=1}^{B_d}$ to train the discriminator, the empirical loss is defined as:
 \begin{align}
\label{eq:algo-defLval}
 \hat L_{Val}(\theta,\omega; \underline{s},\underline{x})=\hat L_{HJB}(\theta,\omega; \underline{s},\underline{x})+\beta_{Val}\hat L_{term}(\theta; \underline{s},\underline{x}),
 \end{align}
 with the HJB PDE residual's empirical loss and the terminal condition's empirical loss being defined as: 
 $$\begin{aligned}
 \hat L_{HJB}&=\frac{1}{B_g}\sum_{j=1}^{B_g}\left|\partial_su_\theta(s_j,x_j)+\frac{\sigma^2}{2}\Delta_xu_\theta(s_j,x_j)+H_\omega\left(s_j,x_j,\nabla_xu_\theta(s_j,x_j)\right)\right|^2,\\
 \hat L_{term}&=\frac{1}{B_g}\sum_{j=1}^{B_g}u_\theta(T,x_j)^2,
 \end{aligned}
 $$
 where $\beta_{Val}>0$ is the weight for the penalty on the terminal condition of $u$.

Note that Algorithm \ref{alg:mfgan-dyn} can be adapted for broader classes of dynamical systems with variational structures. Such GANs structures have been exploited in \cite{Yang2018} and \cite{Yang2018a} to synthesize complex systems governed by physical laws. 

Note also there are various ways to recast MFGs as GANs. For instance, one can swap the roles and view $u$ as a generator (which solves the HJB equation via one neural network) and $m$ as a discriminator (which computes an appropriate differential residual of the FP equation via another neural network).

 
\begin{algorithm}
\caption{MFGANs}
\label{alg:mfgan-dyn}
\begin{algorithmic}
 \STATE{\textbf{Input:} Initial parameters $\theta_0,\omega_0$. Learning rates $\alpha_d,\alpha_g>0$. Number of training samples $B_d$ and $B_g$ for FP and HJB residuals. Prior distribution $p_{prior}$. Numbers $N_{\theta}$ and $N_\omega$ of training steps of the inner-loops and number $K$ of outer-loops} 
 \STATE{\textbf{Output: } Parameters $\theta,\omega$ such that $u_\theta, m_\omega$ approximately solve the MFG PDE system}  
\STATE{Initialize $\theta \leftarrow \theta_0$ and $\omega \leftarrow \omega_0$}
\FOR{$k\in\{1,\dots, K\}$}
 \FOR{$m\in\{1,\dots, N_\omega\}$}
    \STATE{Sample $(\underline{s},\underline{x})=\{(s_i,x_i)\}_{i=1}^{B_d}$ on $[0,T]\times \mathbb R^d$ according to $p_{prior}$ }
    \STATE{$\omega\leftarrow \omega-\alpha_d\nabla_\omega\hat L_{MF}(\theta,\omega; \underline{s},\underline{x})$ where $\hat L_{MF}$ is defined by~\eqref{eq:algo-defLmf}}
 \ENDFOR
 \FOR{$n\in\{1,\dots,N_\theta\}$}
    \STATE{Sample $(\underline{s},\underline{x})=\{(s_j,x_j)\}_{j=1}^{B_g}$ on $[0,T]\times \mathbb R^d$ according to $p_{prior}$ }
    \STATE{$\theta\leftarrow\theta-\alpha_g\nabla_\theta\hat L_{Val}(\theta,\omega; \underline{s},\underline{x})$ where $\hat L_{Val}$ is defined by~\eqref{eq:algo-defLval}}
 \ENDFOR
\ENDFOR
\STATE Return $\theta$, $\omega$
\end{algorithmic}
\end{algorithm}

\subsection{Numerical experiments}
We now assess the quality of the proposed Algorithm \ref{alg:mfgan-dyn}. We start with a class of ergodic MFGs, for both one-dimensional and high-dimensional cases. This class of MFGs is chosen because of their explicit solution structures, which facilitate numerical comparison. We then apply the algorithm to  a different class of high-dimensional MFGs for which there are no known explicit solutions. 

\subsubsection{A class of ergodic MFGs with explicit solutions} \label{subsec:model}
Consider \eqref{eq:mfg-eg} with the following long-run average cost,
\begin{equation}
 \label{eq:cost}
 \hat J_m(\alpha)=\liminf_{T\to\infty}\frac{1}{T}\EE\left[\int_0^T L(X^\alpha_t, \alpha(t,X^\alpha_t)) + f(X^\alpha_t, m(X^\alpha_t)) dt\right],
\end{equation}
where the cost of control and running cost are given by
\[L(x,\alpha) = \frac{1}{2} |\alpha|^2 + \tilde f(x),
	\quad
	f(x, m) = \ln(m), 
	\]
with
\[\tilde f(x)= 2 \pi^2 \left[ - \sum_{i=1}^d \sin(2 \pi x_i) + \sum_{i=1}^d |\cos(2 \pi x_i)|^2 \right]- 2\sum_{i=1}^d \sin(2 \pi x_i).\]
In this ergodic setting, the PDE system \eqref{eq:hjb}--\eqref{eq:fp} becomes
\begin{equation}\label{eq:hjb-fp}
 \begin{cases}
 -\epsilon\Delta u+H_0(x,\nabla u)=f(x,m)+\bar H,\\
 -\epsilon\Delta m-\div\left(m\nabla_pH_0(x,\nabla u)\right)=0,\\
 \int_{\mathbb T^d}u(x)dx=0;\,m>0,\, \int_{\mathbb T^d}m(x)dx=1,
 \end{cases}
\end{equation}
where the convex conjugate $H_0$ of $L$ defined by~\eqref{eq:defH0} is given by 
$H_0(x,p)=\sup_{\alpha}\{\alpha\cdot p-\frac{1}{2}|\alpha|^2\}-\tilde f(x).$ 
Here, the periodic value function $u$, the periodic density function $m$, and the unknown ergodic constant $\bar H$ can be explicitly derived as discussed in~\cite{MR3698446}. Indeed, assuming the existence of a smooth solution $(m,u,\bar H)$, $m$ in the second equation in \eqref{eq:hjb-fp} can be written as $m(x) = \frac{e^{2u(x)}}{\int_{\mathbb T^d} e^{2u(x')}dx'}.$ 
Hence the solution to \eqref{eq:hjb-fp} is given by $u(x) = \sum_{i=1}^d \sin(2 \pi x_i)$ and the ergodic constant is $\bar H = \ln\left(\int_{\mathbb T^d} e^{2 \sum_{i=1}^d \sin(2 \pi x_i)} d x\right).$
The optimal control policy is also explicitly given by
$$\begin{aligned}
\alpha^*(x)
    &=\arg\max_{\alpha \in \RR^d}\{\nabla_xu(x)\cdot \alpha -L(x,\alpha)\}
    =\nabla_xu(x)=2\pi\begin{pmatrix*}\cos(2\pi x_1)&\dots&\cos(2\pi x_d)\end{pmatrix*}\in\RR^d.
\end{aligned}$$

\begin{figure}[!ht]
 \centering
 \begin{subfigure}[b]{0.4\columnwidth}
 \centering
 \includegraphics[width=\textwidth]{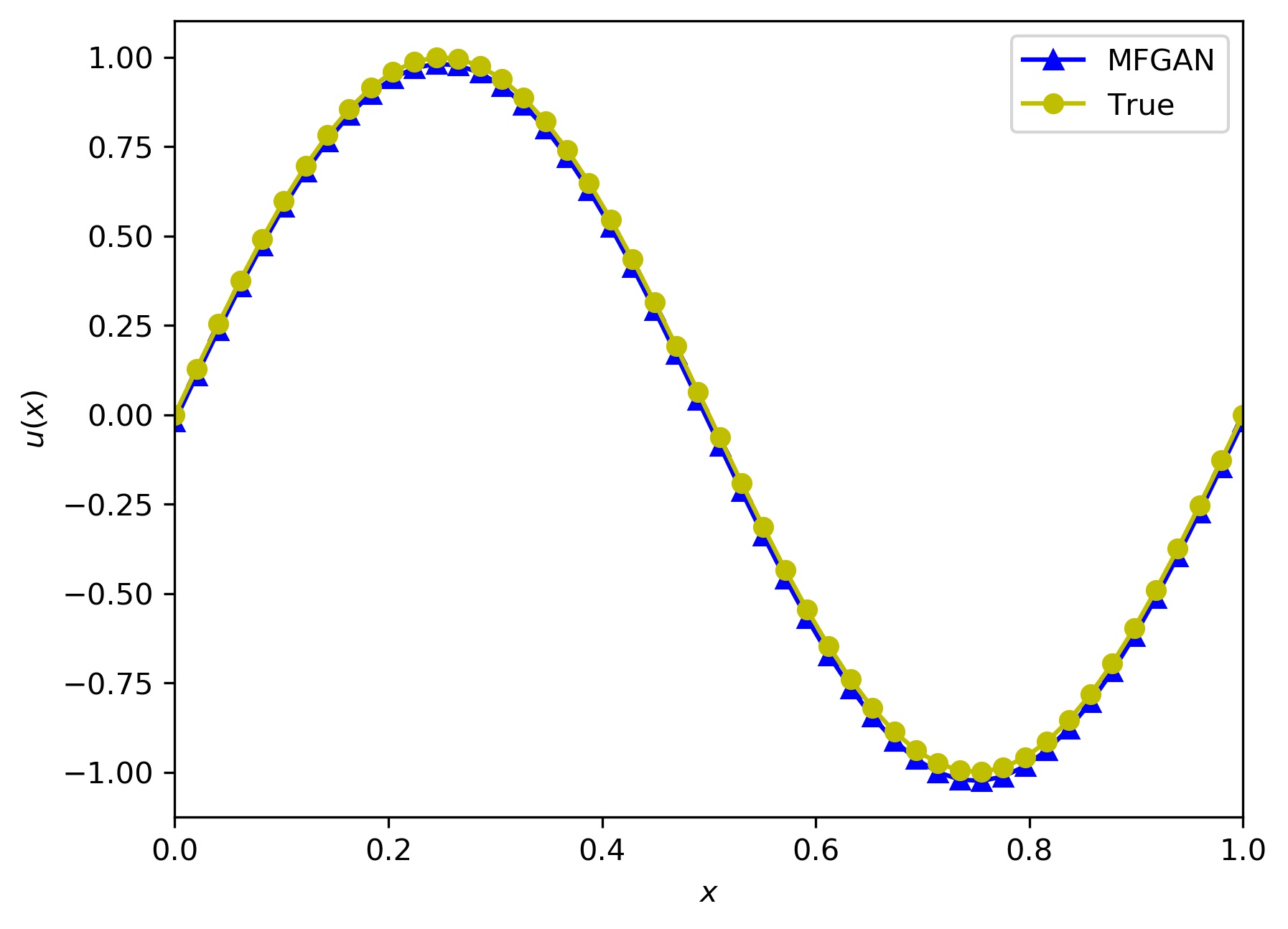}
 \caption{Value function $u$.}
 \label{subfig:supp-1dim-u}
 \end{subfigure}
 \begin{subfigure}[b]{0.4\columnwidth}
 \centering
 \includegraphics[width=\textwidth]{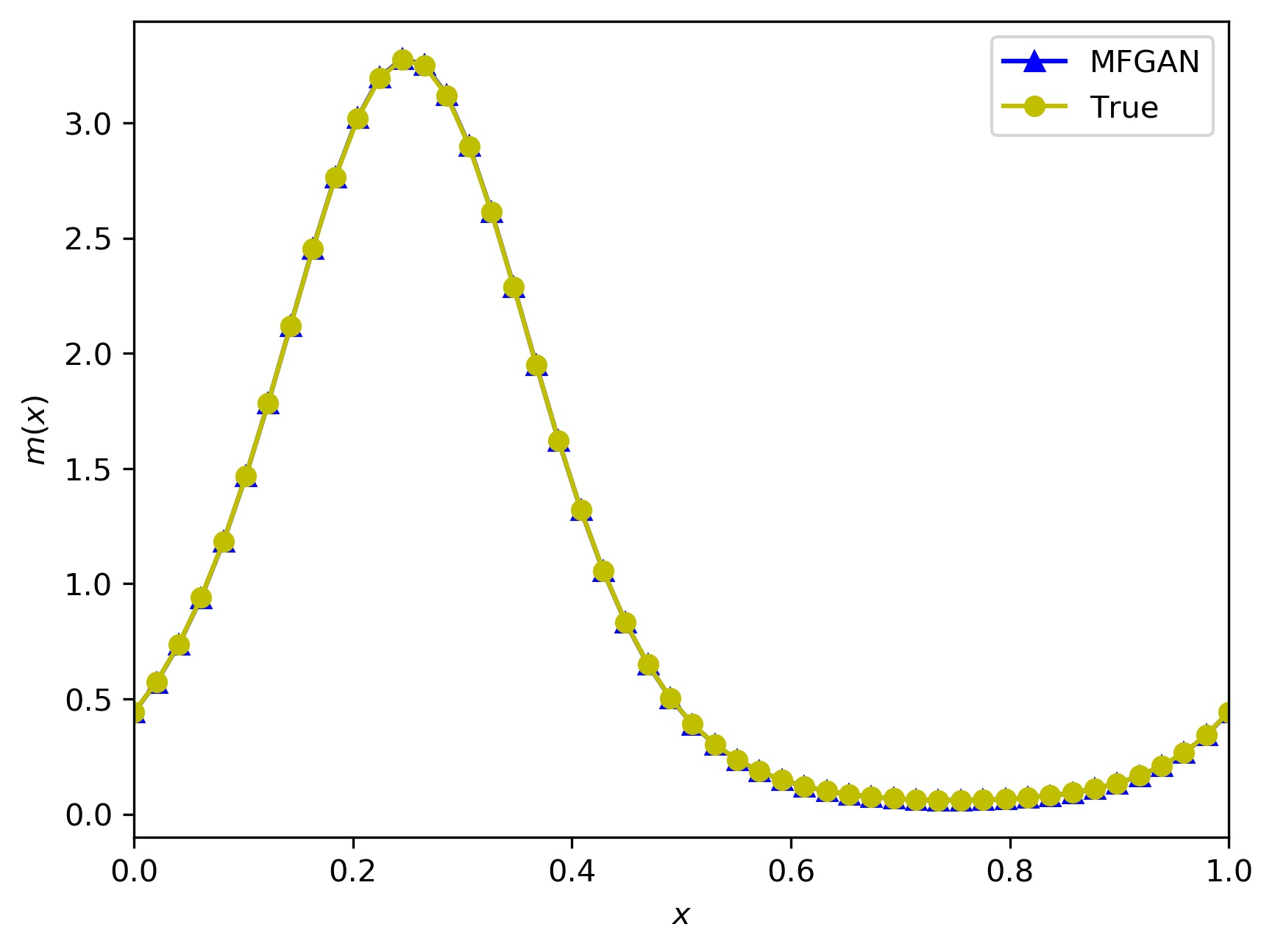}
 \caption{Density function $m$.}
 \label{subfig:supp-1dim-m}
 \end{subfigure}\\
 \begin{subfigure}[t]{0.4\columnwidth}
 \centering
 \includegraphics[width=\textwidth]{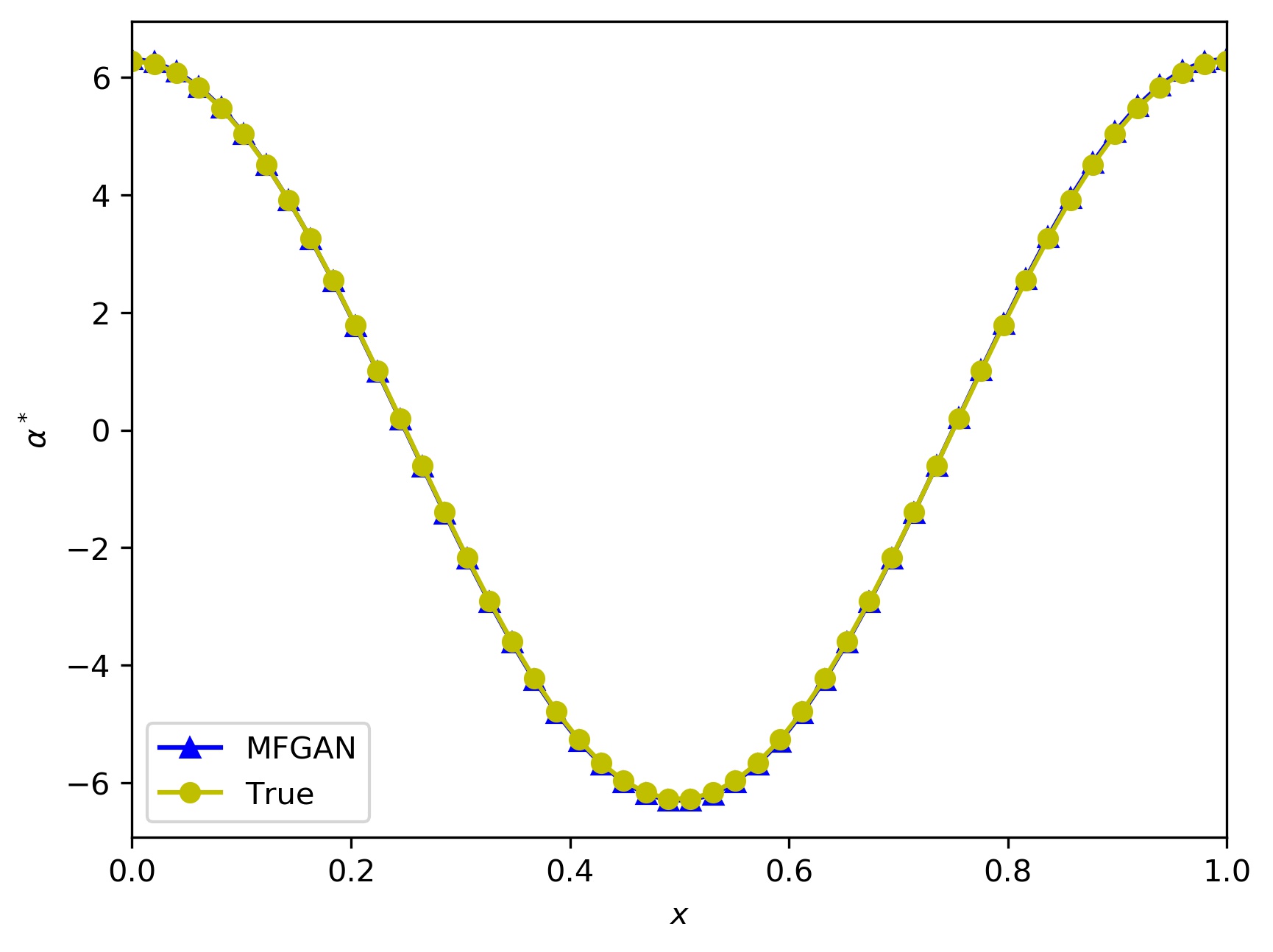}
 \caption{Optimal control $\alpha^*$.}
 \label{subfig:supp-1dim-alpha}
 \end{subfigure}
 \caption{One-dimensional test case.}
 \label{fig:supp-1dim}
\end{figure}
\begin{figure}[!ht]
 \centering
 \begin{subfigure}[b]{0.4\columnwidth}
 \centering
 \includegraphics[width=\textwidth]{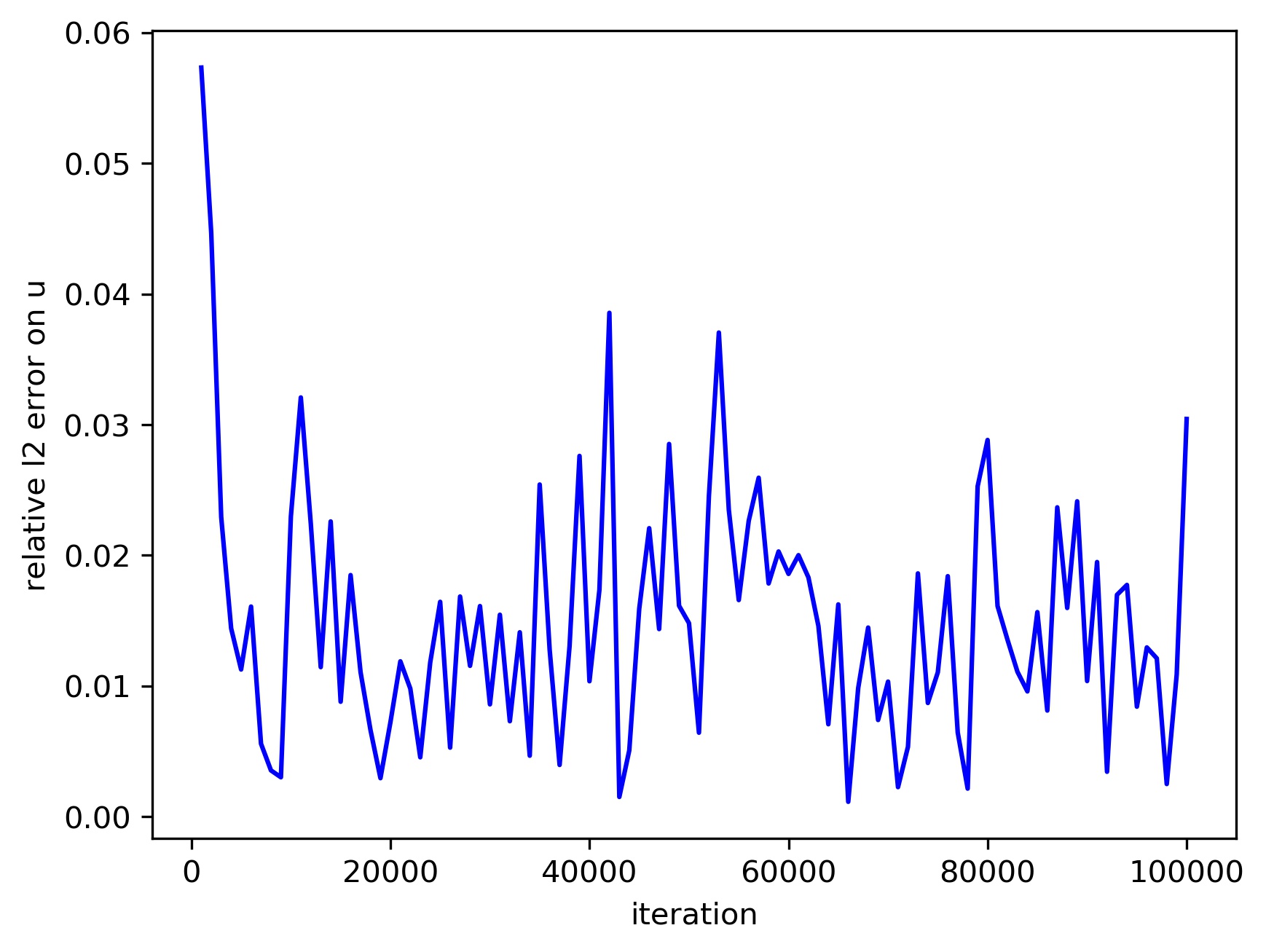}
 \caption{Relative $l_2$ error of $u$.}
 \label{subfig:1dim-rel-err-u}
 \end{subfigure}
 \begin{subfigure}[b]{0.4\columnwidth}
 \centering
 \includegraphics[width=\textwidth]{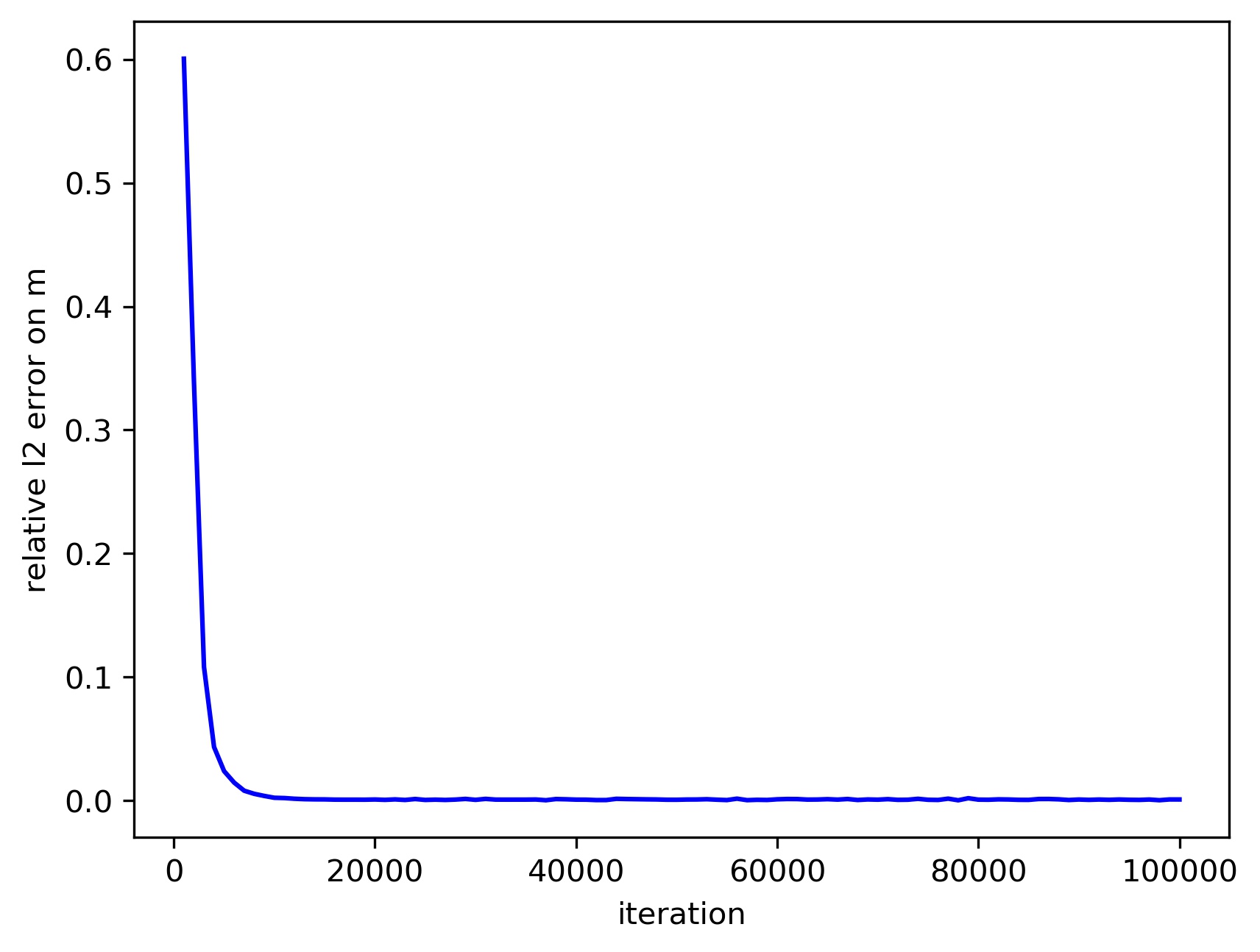}
 \caption{Relative $l_2$ error of $m$.}
 \label{subfig:1dim-rel-err-m}
 \end{subfigure}\\
 \begin{subfigure}[b]{0.4\columnwidth}
 \centering
 \includegraphics[width=\textwidth]{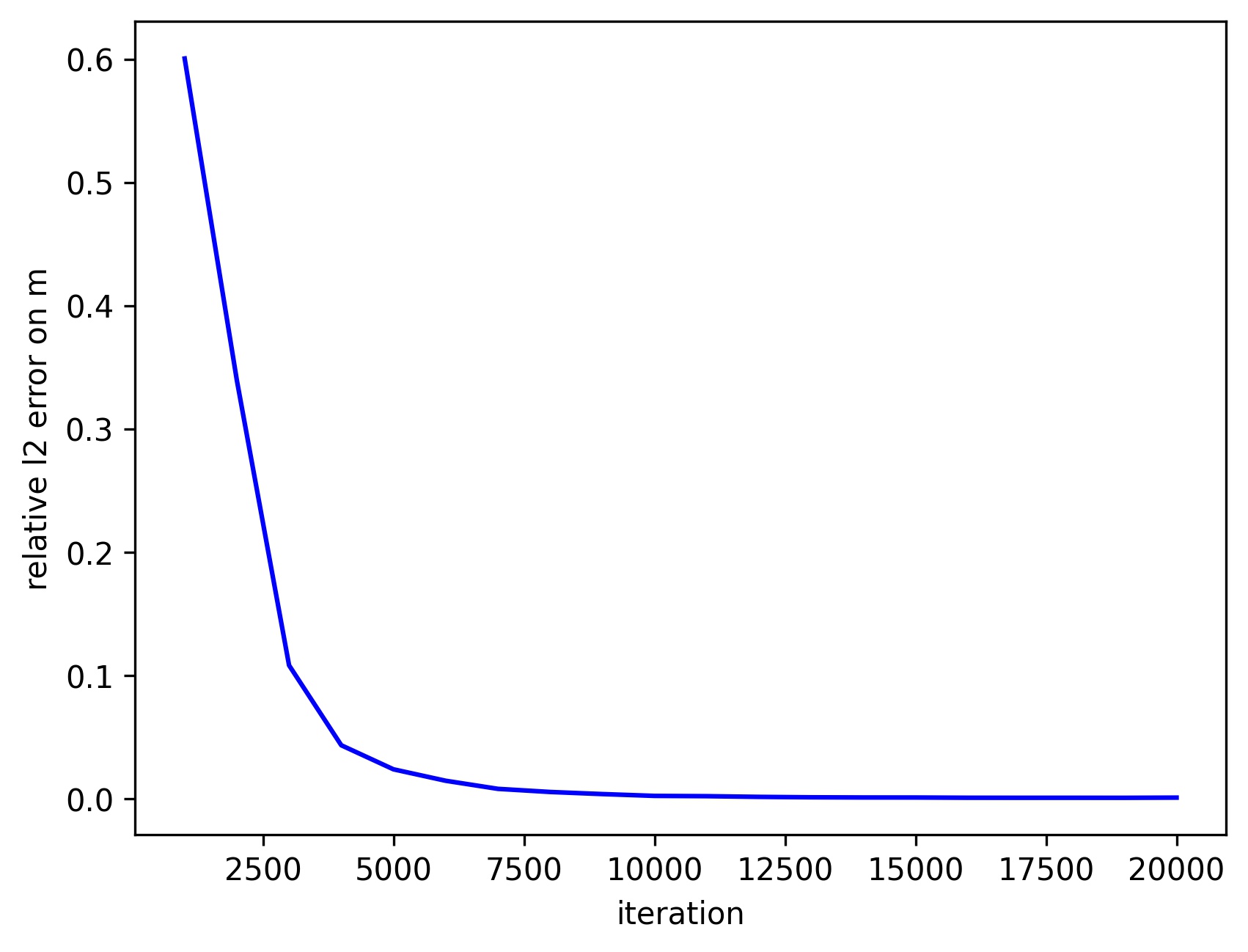}
 \caption{Error of $m$ first 20k iterations.}
 \label{subfig:1dim-rel-err-m-dec}
 \end{subfigure}
 \begin{subfigure}[b]{0.4\columnwidth}
 \centering
 \includegraphics[width=\textwidth]{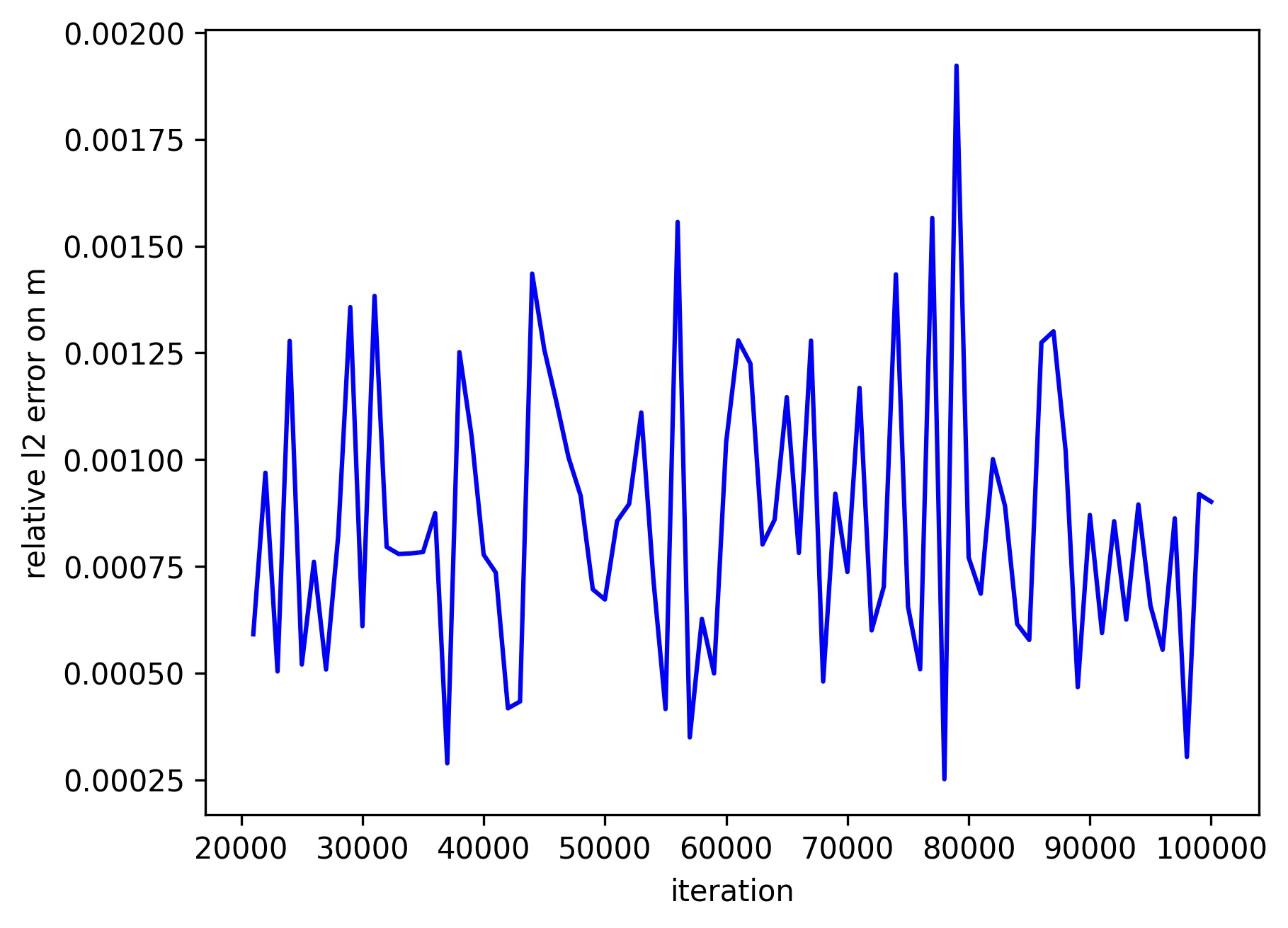}
 \caption{Error of $m$ after 20k iterations.}
 \label{subfig:1dim-rel-err-m-flt}
 \end{subfigure}\\
 \begin{subfigure}[b]{0.4\columnwidth}
 \centering
 \includegraphics[width=\textwidth]{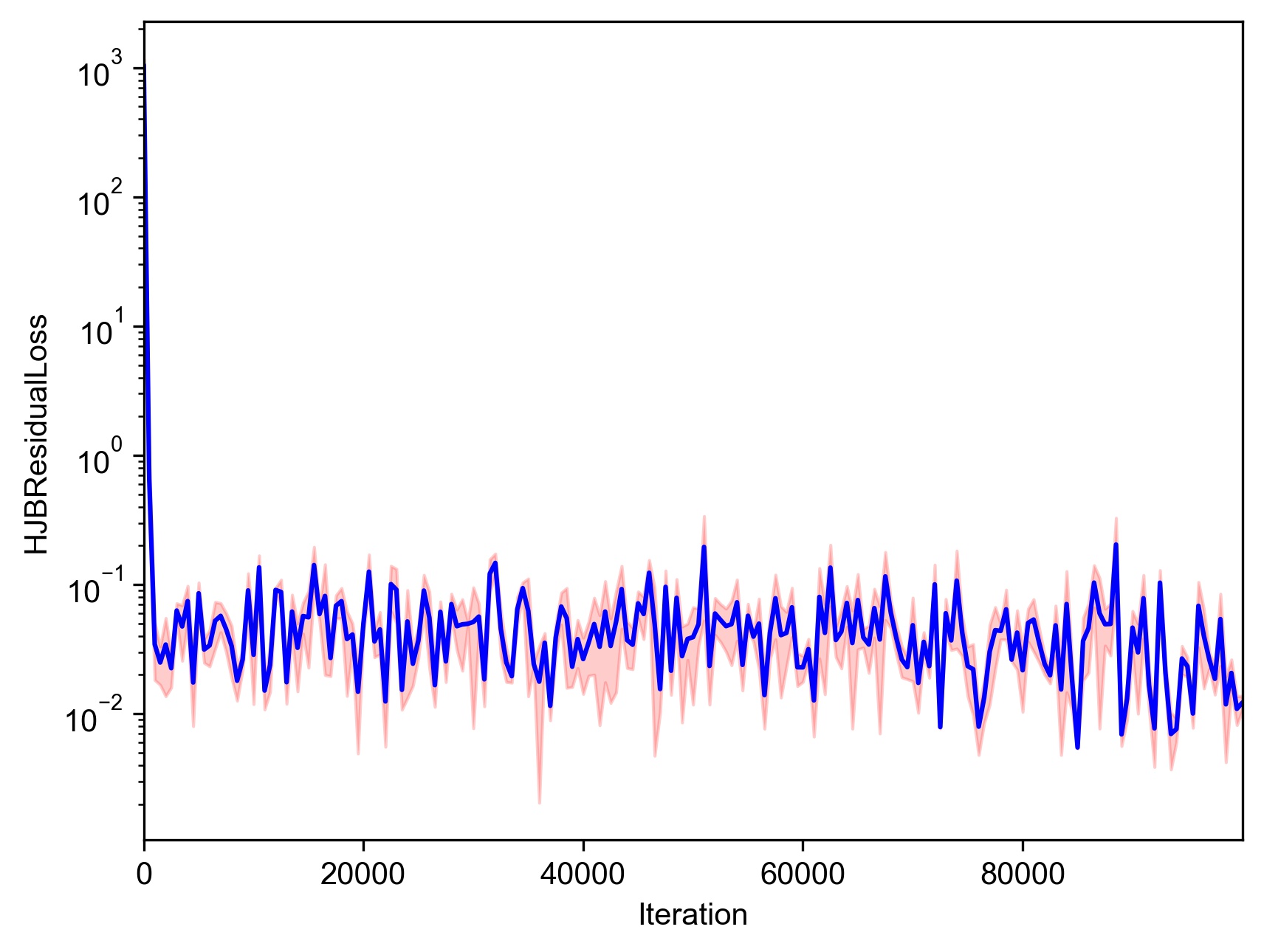}
 \caption{HJB residual loss.}
 \label{subfig:1dim-hjb}
 \end{subfigure}
 \begin{subfigure}[b]{0.4\columnwidth}
 \centering
 \includegraphics[width=\textwidth]{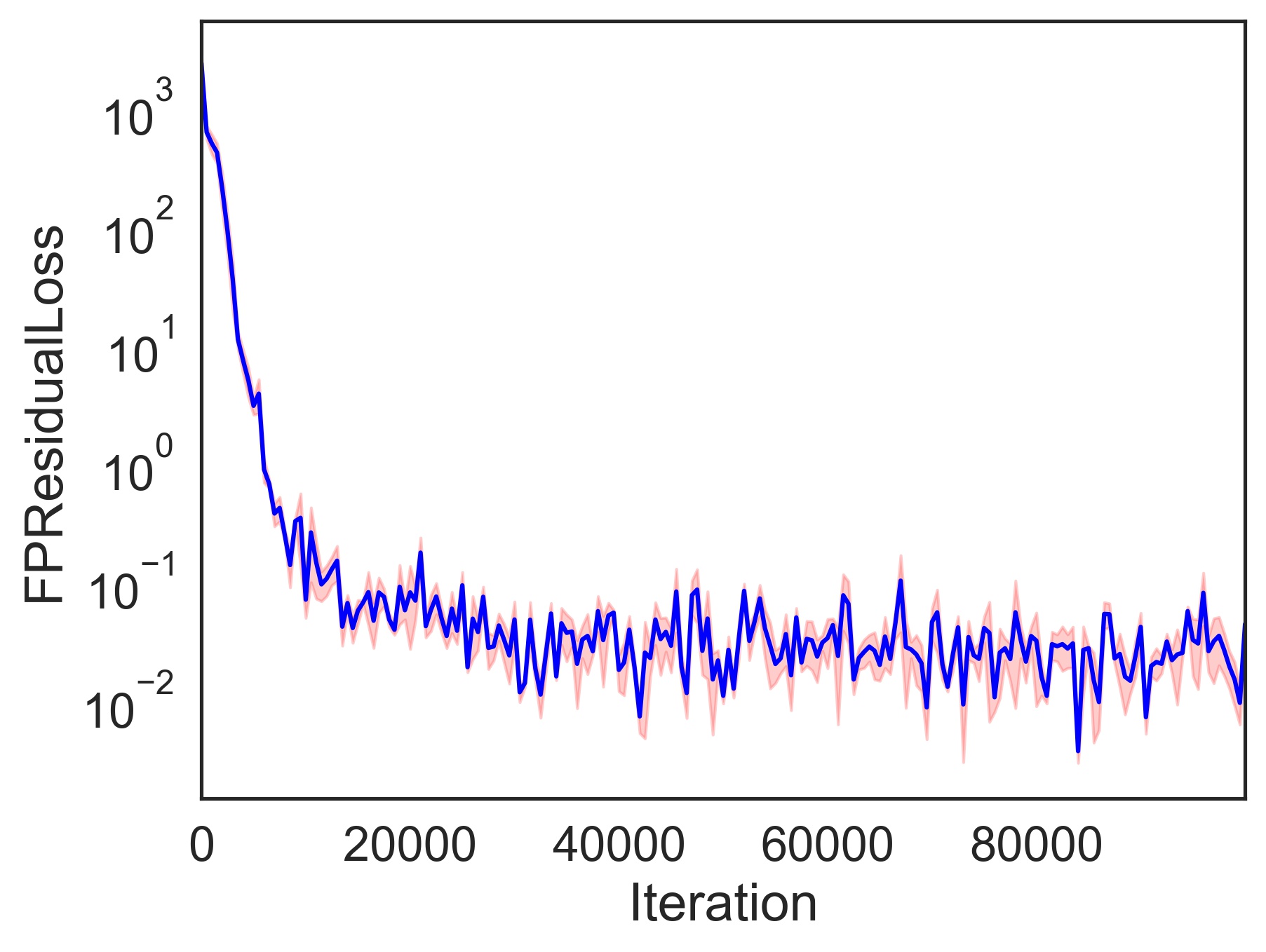}
 \caption{FP residual loss.}
 \label{subfig:1dim-fp}
 \end{subfigure}
 \caption{Losses and errors in the one-dimensional test case.}
 \label{fig:1dim}
\end{figure}

\paragraph{Implementation.}
Both the value function $u$ and the density function $m$ are computed via neural networks with parameters ${\theta}$ and $\omega$ respectively. Moreover, 
\begin{itemize}
 \item 
The neural network $m_\omega$ is assumed to be a maximum entropy probability distribution, i.e., of the form $m_\omega\propto\exp{f_\omega}$ for some neural network $f_\omega$. This is due to the lack of information about the density function $m$. (See also \cite{Finn2016}). The density function $m_\omega$ is normalized and therefore $\beta_{MF}=0$.
\item The network architecture for implementing both $u_\theta$ and $f_\omega$ adopts 
the Deep Galerkin Method (DGM) architecture proposed in \cite{sirignano2018dgm}. The DGM architecture is known to be useful for solving PDEs numerically.
(See for instance \cite{sirignano2018dgm,al2018solving,CarmonaLauriere_DL}). The DGM network for both $u_\theta$ and $f_\omega$ contains 1 hidden layer of DGM type with $4$ nodes. The activation function for $u_\theta$ is hyperbolic tangent function and that of $f_\omega$ is sigmoid function.
\end{itemize} 
Since the MFG is of ergodic type with a specified periodicity, the architecture and Algorithm \ref{alg:mfgan-dyn} are adapted accordingly. More precisely, 
\begin{itemize}
 \item To accommodate the periodicity given by the domain flat torus $\mathbb T^d$, for any data point ${x_i=(x_{i,1},\dots,x_{i,d})\in\mathbb R^d}$, we use
 \begin{equation}\label{eq:input-fourier}\begin{aligned}
 y_i=&\left(\sin{(2\pi x_{i,1})},\dots,\sin{(2\pi x_{i,d})},\right.\\
 &\hspace{1pt}\left.\cos{(2\pi x_{i,1})},\dots,\cos{(2\pi x_{i,d})}\right)
 \end{aligned}
 \end{equation}
 as input. The $x_i's$ and $y_i's$ here are the latent variables in the vanilla GANs.
 \item An additional trainable variable $\bar H$ is introduced in the graphical model.
 \item The loss functions $\hat L_{HJB}$ and $\hat L_{FP}$ are modified according to the first and second equations of \eqref{eq:hjb-fp}. The generator penalty becomes
 \[\hat L_{term}=\left[\frac{\sum_{i=1}^{B_g}u_\theta(y_i)}{B_g}\right]^2.\]
 For the weight for the generator penalty, we used $\beta_{Val}=1$.
 \item As GANs can be trained in an alternative fashion between the generator and discriminator, we make the following changes in this experiment to suit our training objective: Within each iteration, we first update the generator by $N_\theta=5$ stochastic gradient descent (SGD) steps with initial learning rate $1\times10^{-3}$ and then update of the discriminator by $N_\omega=2$ SGD steps with initial learning rate $1\times10^{-4}$. The minibatch sizes are set to be $B_g = B_d = 32$. The number of total iterations, i.e., the number of outer loops is $K=10^5$. Adam optimizer is used for the updates.
\end{itemize}

\begin{figure}[!ht]
 \centering
 \begin{subfigure}[b]{0.4\columnwidth}
 \centering
 \includegraphics[width=\textwidth]{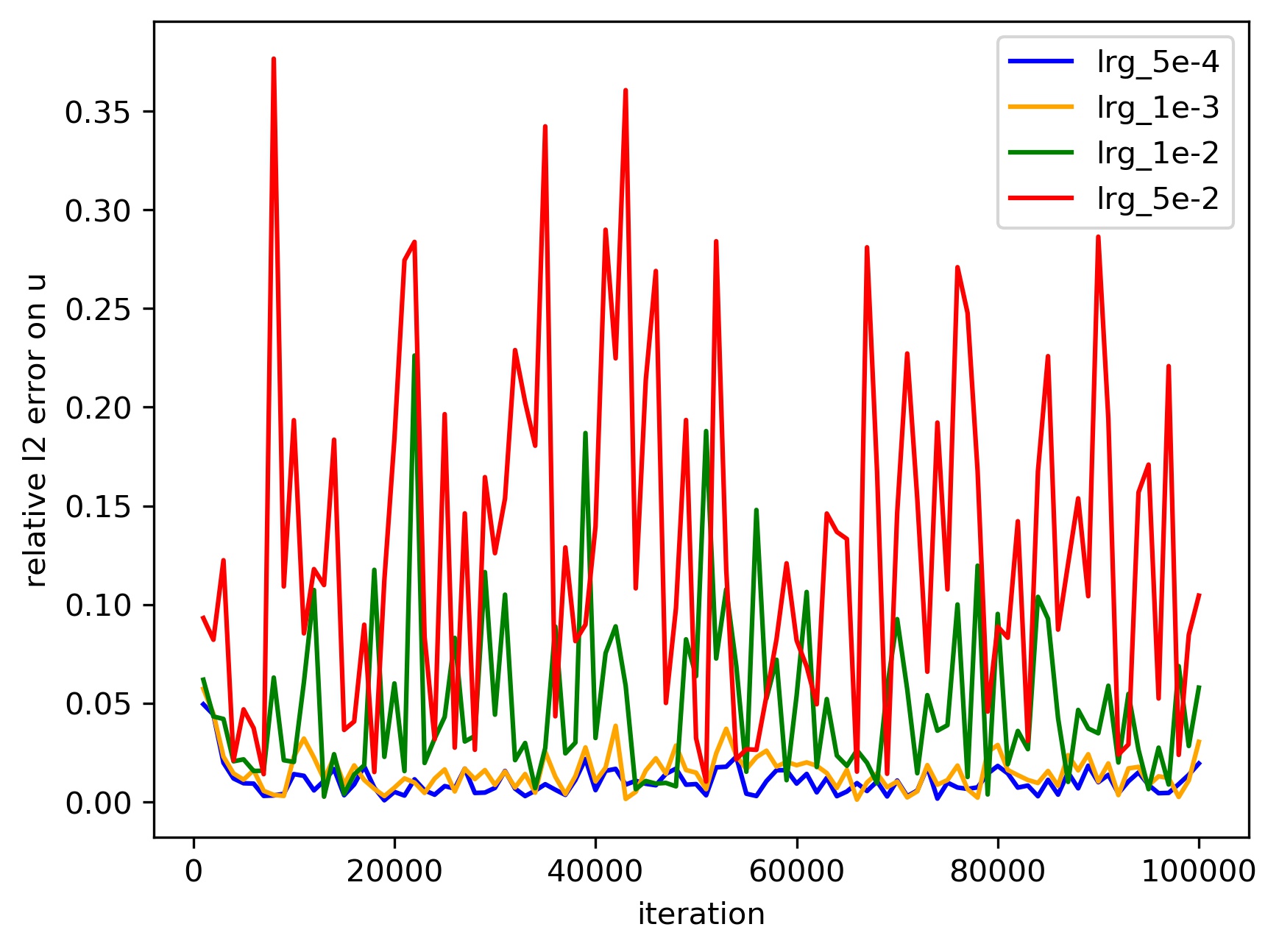}
 \caption{Relative $l_2$ error of $u$.}
 \label{subfig:ab-lrg-rel-err-u}
 \end{subfigure}
 \begin{subfigure}[b]{0.4\columnwidth}
 \centering
 \includegraphics[width=\textwidth]{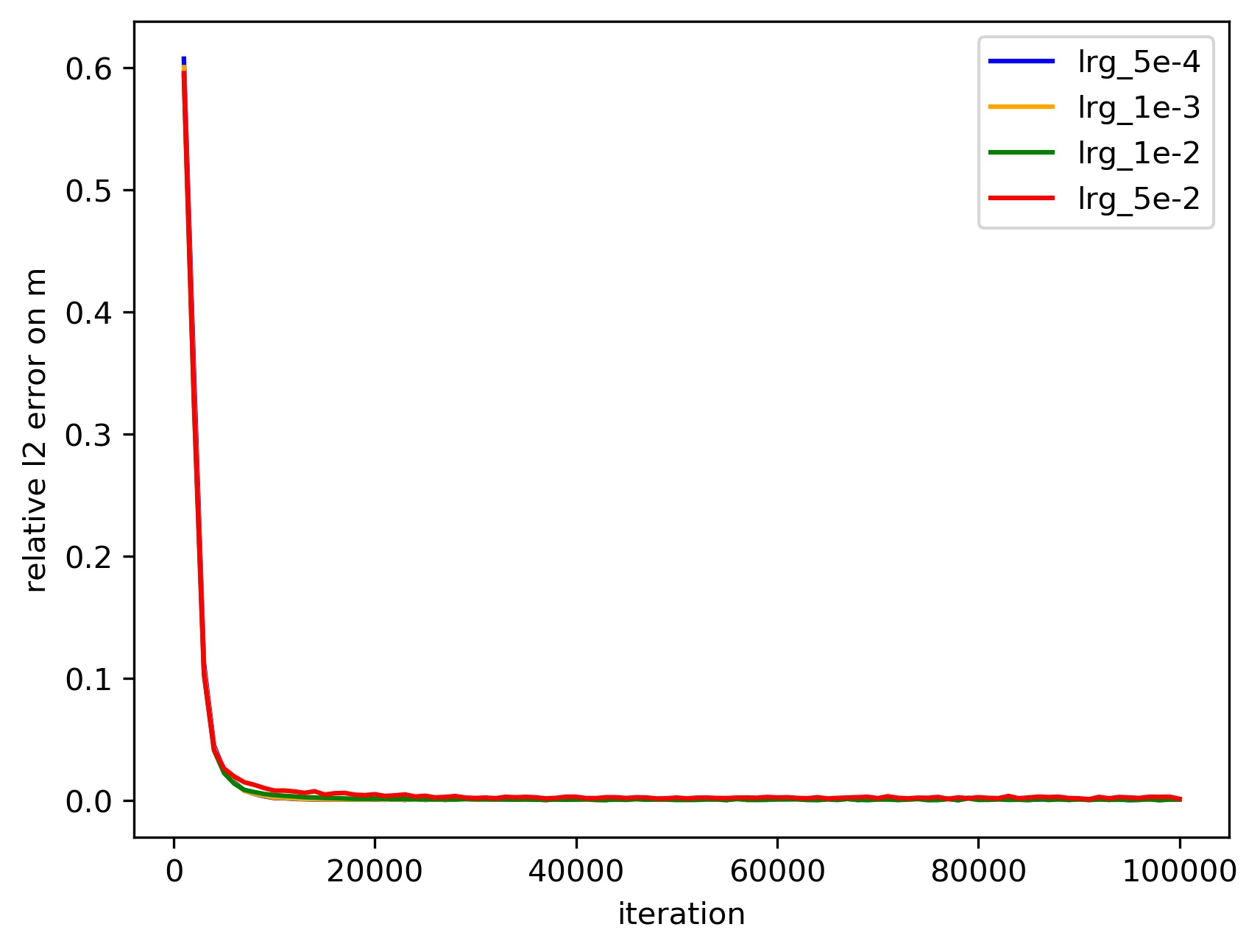}
 \caption{Relative $l_2$ error of $m$.}
 \label{subfig:ab-lrg-rel-err-m}
 \end{subfigure}
 \caption{Impact of generator learning rate on relative $l_2$ error.}
 \label{fig:ab-lrg-rel-err}
\end{figure}

\paragraph{Performance evaluations.}
To assess the performance of our algorithm, the following procedure is adopted.
\begin{itemize}
 \item Given the explicit solution to the MFG \eqref{eq:hjb-fp}, we compare the learnt value function, the learnt density function and the learnt optimal control against their respective analytical form.
 \item We adopt the evolution of relative $l_2$ errors between the learnt and true value and density functions. The relative $l_2$ error of a function $f$ against another function $g$, with $f,g:\mathbb{T}^d\to\mathbb R$ and $g$ not constant 0, is given by
 \[err_{rel-l_2}(f,g)=\sqrt{\frac{\int_{\mathbb{T}^d}[f(x)-g(x)]^2dx}{\int_{\mathbb T^d}g(x)^2dx}}.\]
 \end{itemize}
Moreover, to facilitate comparisons for broader classes of MFGs whose analytical solutions may not be available, additional loss functions are adopted. Here we take 
differential residuals of both the HJB and the FP equations as measurement of the performance.

\begin{figure}[!ht]
 \centering
 \begin{subfigure}[b]{0.4\columnwidth}
 \centering
 \includegraphics[width=\textwidth]{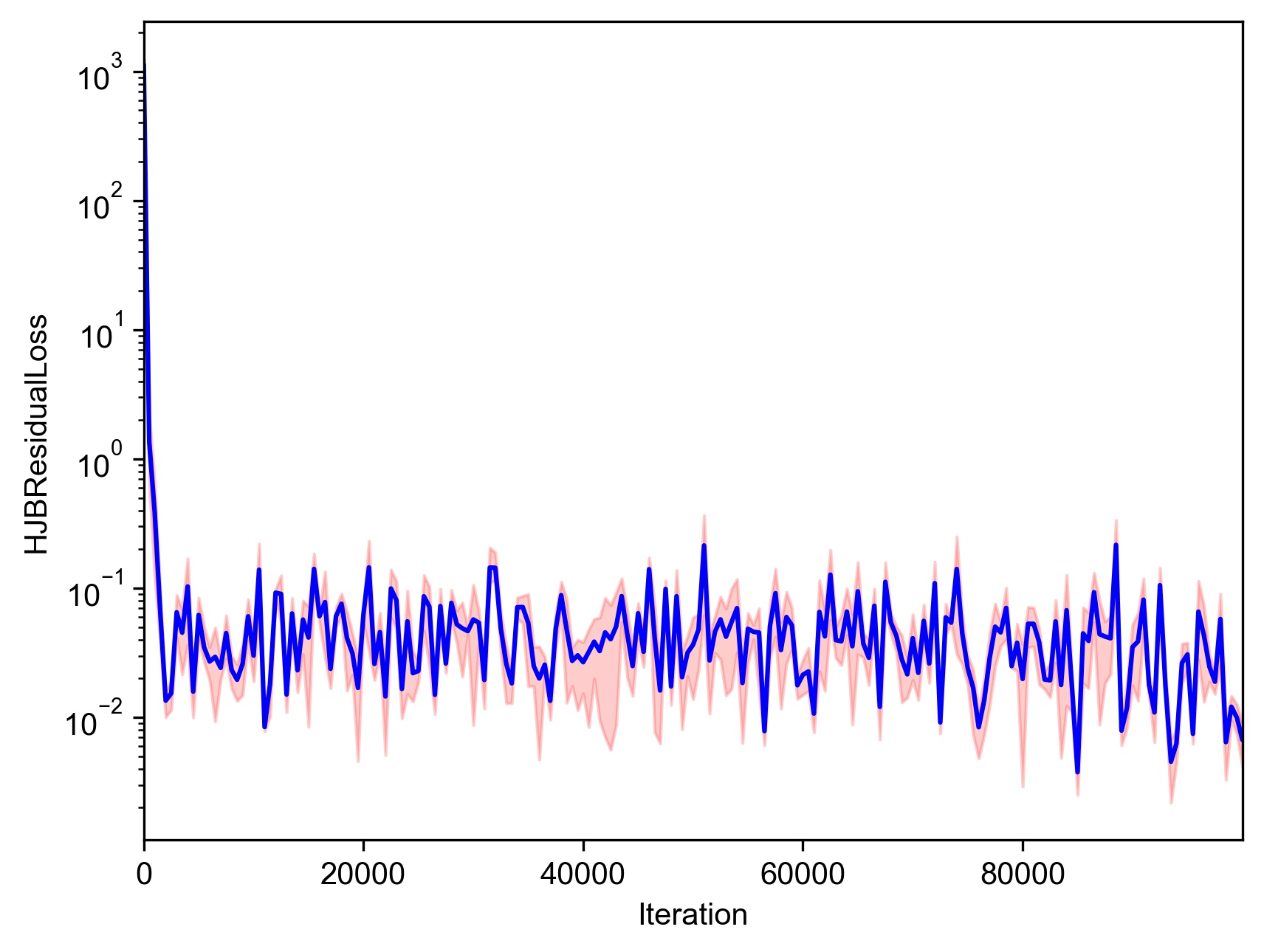}
 \caption{$\alpha_g^0=5\times10^{-4}$.}
 \label{subfig:hjb-smlrg}
 \end{subfigure}
 \begin{subfigure}[b]{0.4\columnwidth}
 \centering
 \includegraphics[width=\textwidth]{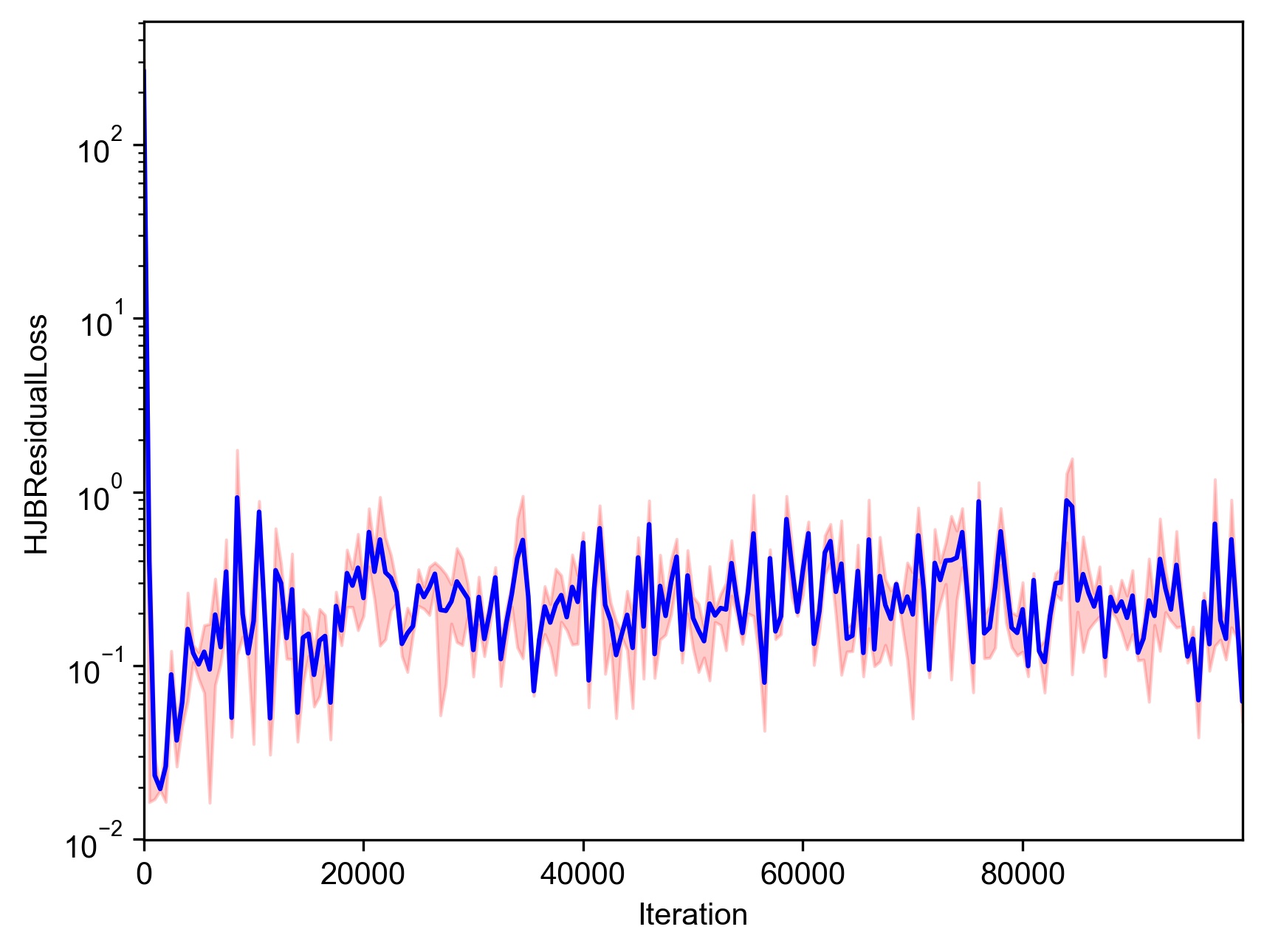}
 \caption{$\alpha_g^0=5\times10^{-2}$.}
 \label{subfig:hjb-lglrg}
 \end{subfigure}
 \caption{HJB residual loss under different generator learning rate.}
 \label{fig:hjb-lrg}
\end{figure}
\begin{figure}[!ht]
 \centering
 \begin{subfigure}[b]{0.4\columnwidth}
 \centering
 \includegraphics[width=\textwidth]{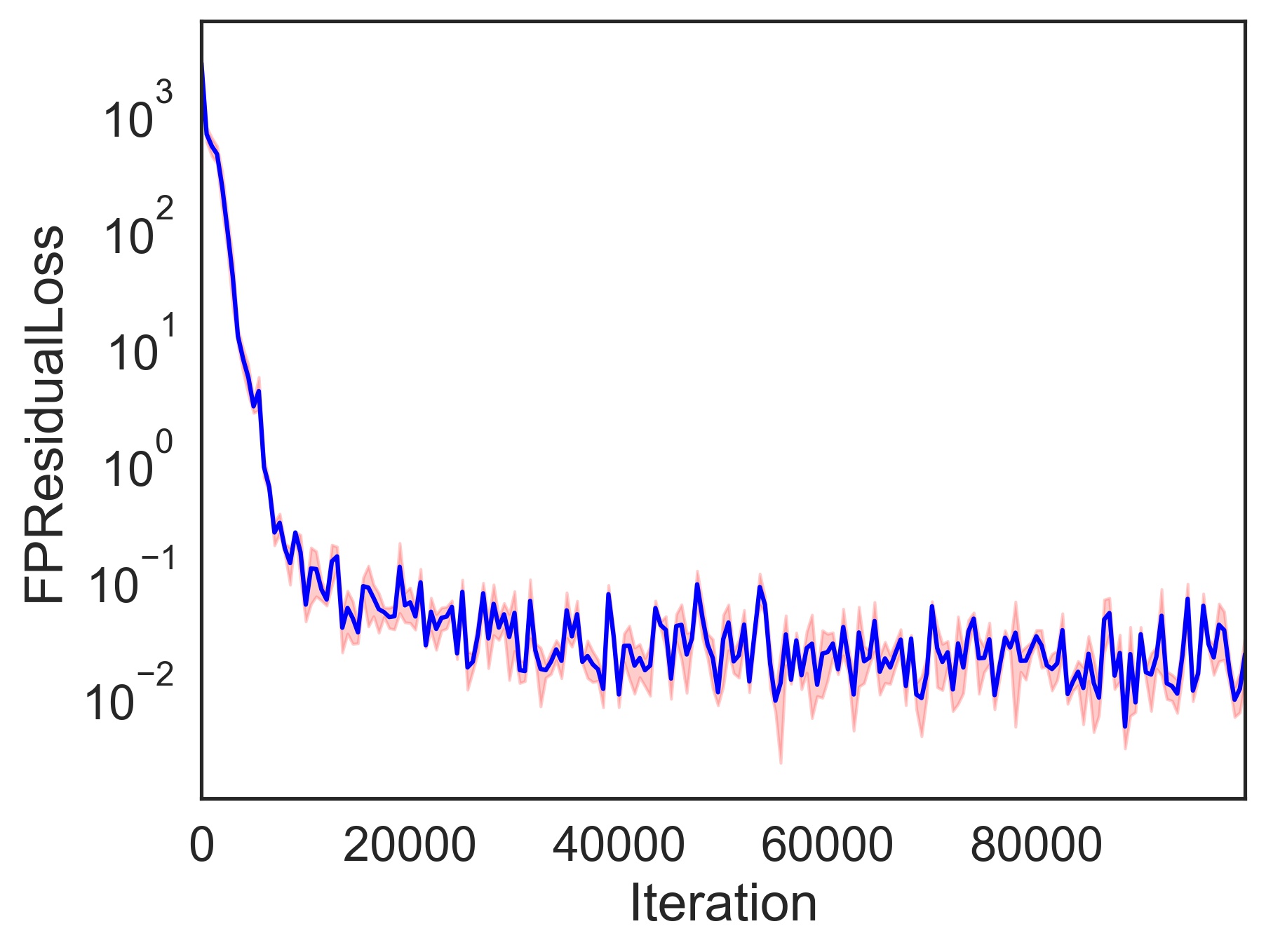}
 \caption{$\alpha_g^0=5\times10^{-4}$.}
 \label{subfig:fp-smlrg}
 \end{subfigure}
 \begin{subfigure}[b]{0.4\columnwidth}
 \centering
 \includegraphics[width=\textwidth]{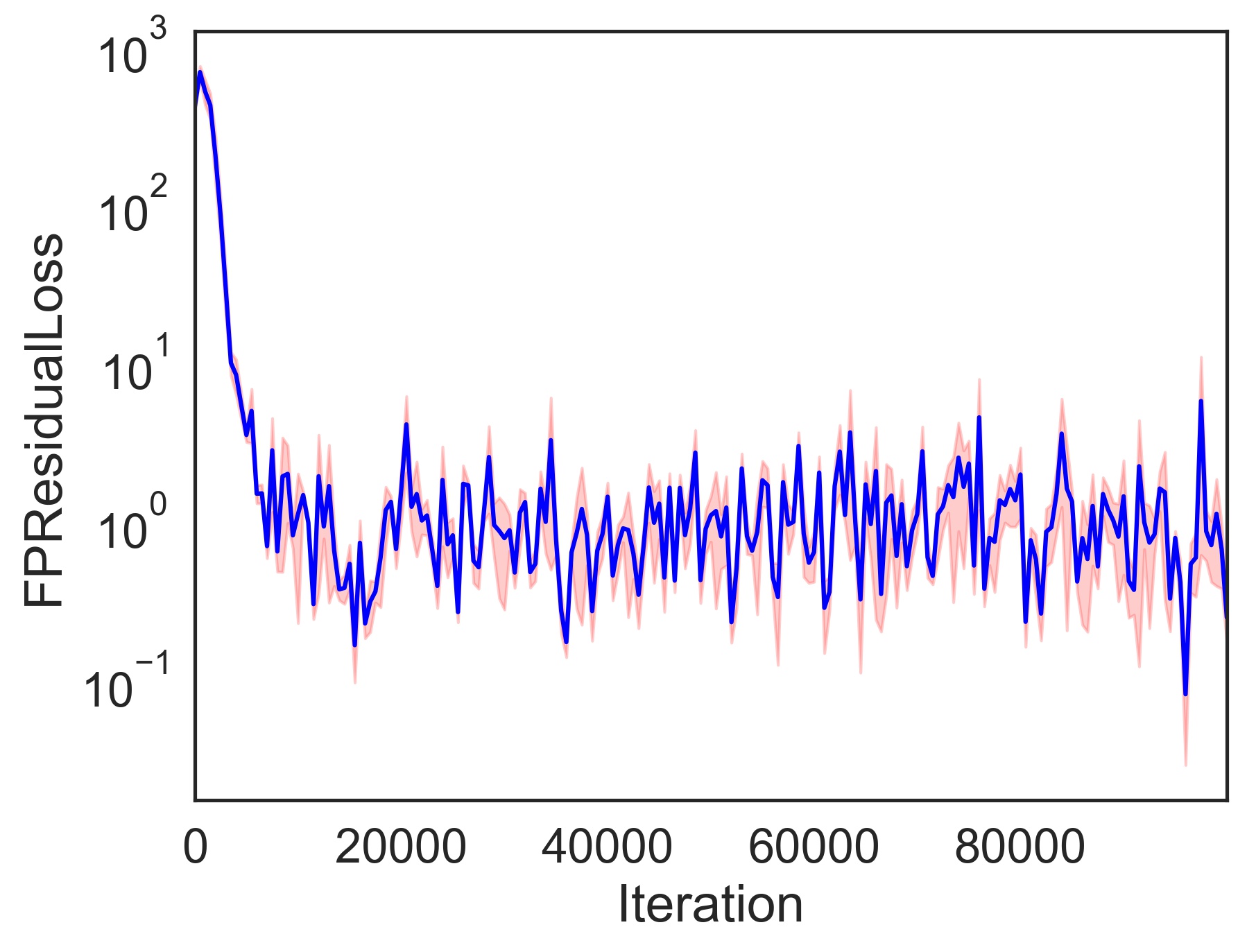}
 \caption{$\alpha_g^0=5\times10^{-2}$.}
 \label{subfig:fp-lglrg}
 \end{subfigure}
 \caption{FP residual loss under different generator learning rate.}
 \label{fig:fp-lrg}
\end{figure}

\begin{figure}[!ht]
 \centering
 \begin{subfigure}[b]{0.4\columnwidth}
 \centering
 \includegraphics[width=\textwidth]{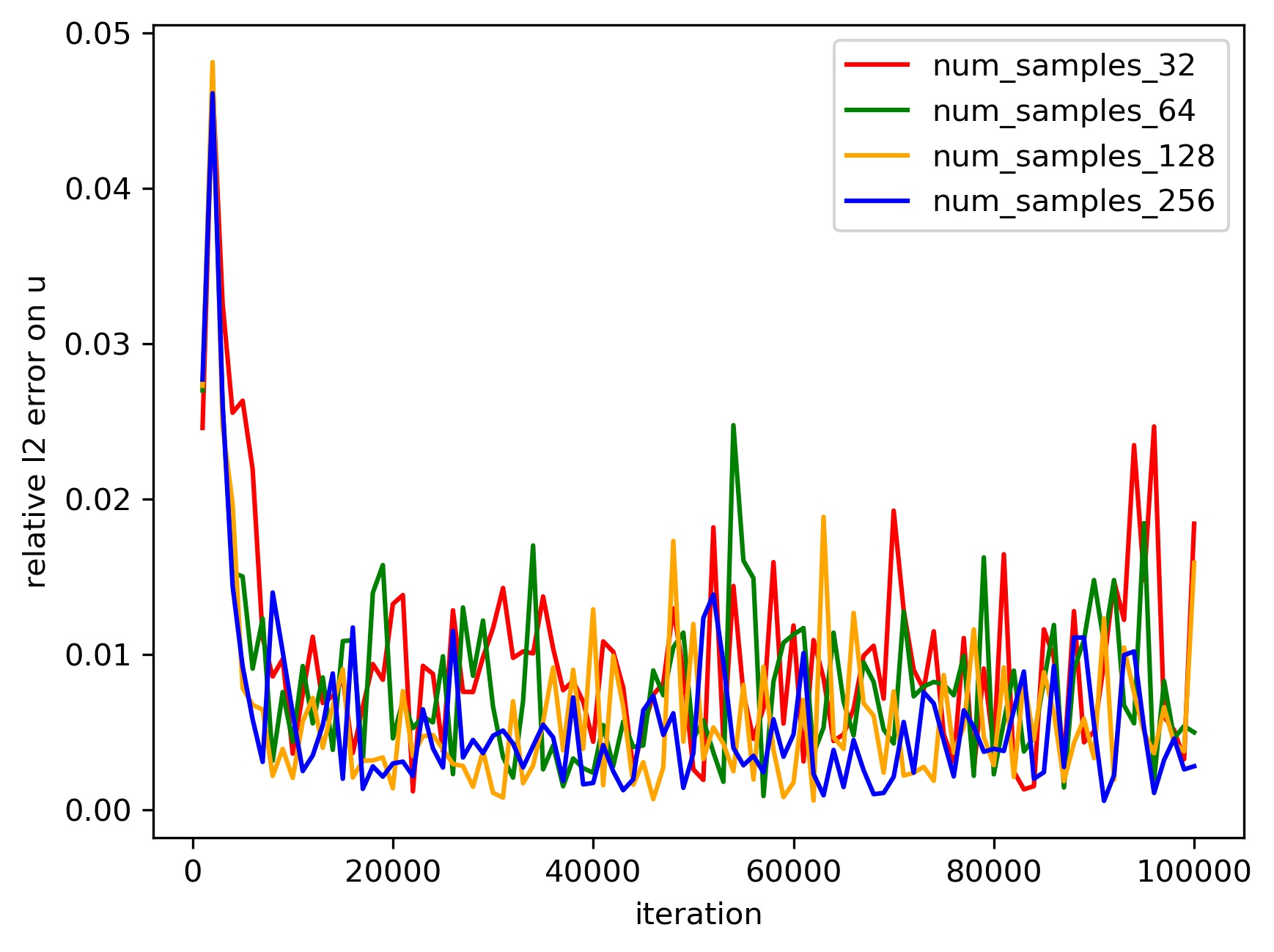}
 \caption{Relative $l_2$ error of $u$.}
 \label{subfig:ab-num-rel-err-u}
 \end{subfigure}
 \begin{subfigure}[b]{0.4\columnwidth}
 \centering
 \includegraphics[width=\textwidth]{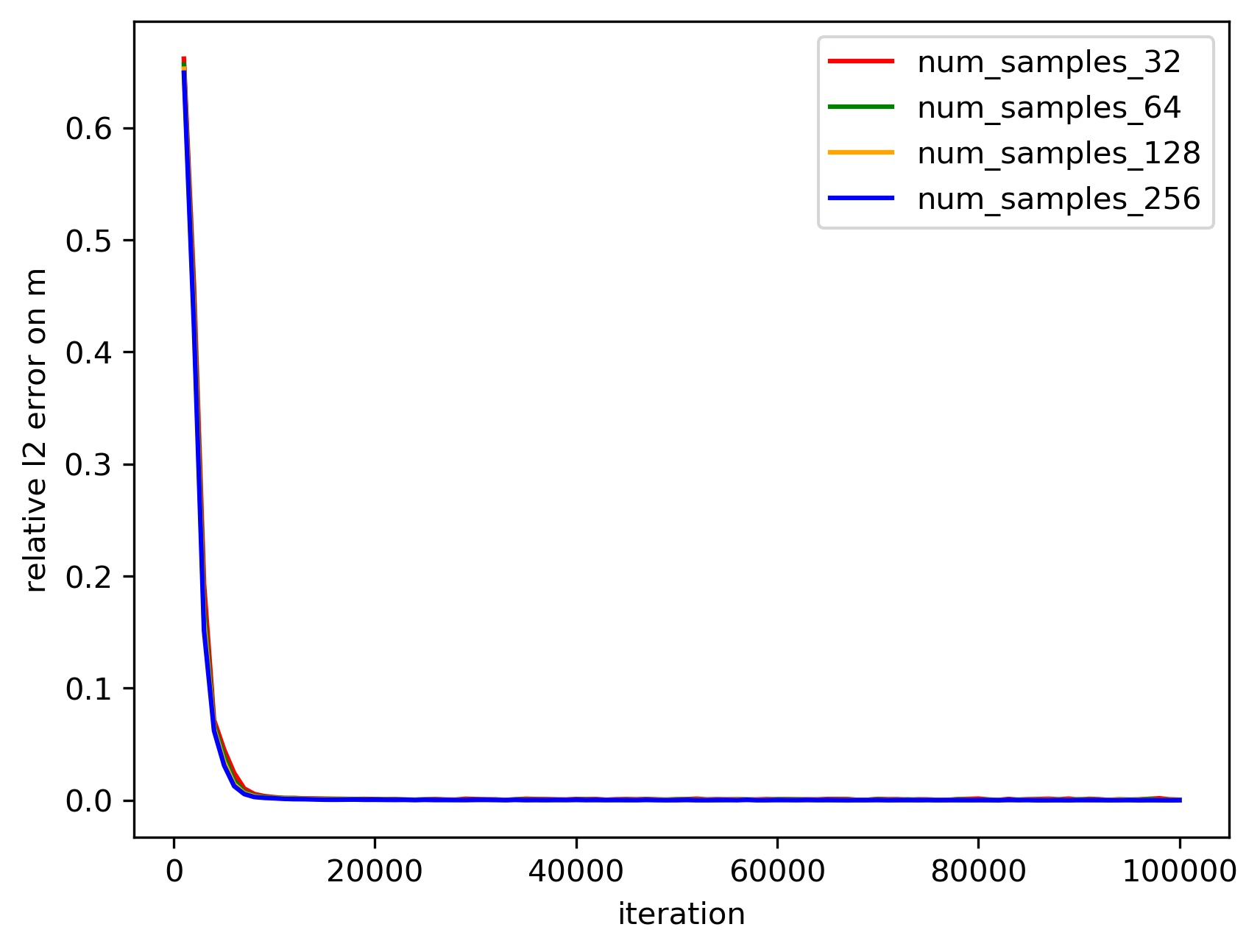}
 \caption{Relative $l_2$ error of $m$.}
 \label{subfig:ab-num-rel-err-m}
 \end{subfigure}
 \caption{Impact of minibatch size on relative $l_2$ errors.}
 \label{fig:ab-num-rel-err}
\end{figure}

The results are presented in Figures \ref{fig:supp-1dim} and \ref{fig:1dim}. Figures \ref{subfig:supp-1dim-u} and \ref{subfig:supp-1dim-m} show the learnt functions of $u$ and $m$ against the true ones, respectively, and Figure \ref{subfig:supp-1dim-alpha} shows the optimal control. Both figures demonstrate the accuracy of the learnt functions versus the true ones. This strong performance is supported by the plots of loss in Figures \ref{subfig:1dim-rel-err-u} and \ref{subfig:1dim-rel-err-m}, depicting the evolution of relative $l_2$ error as the number of outer iterations grows to $K$. Within $10^5$ iterations, the relative $l_2$ error of $u$ oscillates around $3\times10^{-2}$, and the relative $l_2$ errors of $m$ decreases below $10^{-3}$. 

The evolution of the HJB and FP differential residual loss (see $L_{HJB}$ and $L_{FP}$) is shown in Figures \ref{subfig:1dim-hjb} and \ref{subfig:1dim-fp}, respectively. In these figures, the solid line is the average loss over $3$ experiments, with standard deviation captured by the shaded region around the line. Both differential residuals first rapidly descend to the magnitude of $10^{-2}$; then the descent slows down and is accompanied by oscillation. 

One may notice the difference between the training results of $u$ and $m$: one reason being that $u$ and $m$ are implemented using different neural networks; the other being that different loss functions are adopted for training $u$ and $m$.

\begin{figure}[!ht]
 \centering
 \begin{subfigure}[b]{0.4\columnwidth}
 \centering
 \includegraphics[width=\textwidth]{figures/ablation_study/HJBResidualLoss.jpg}
 \caption{$B_g=B_d=32$.}
 \label{subfig:hjb-smnum}
 \end{subfigure}
 \begin{subfigure}[b]{0.4\columnwidth}
 \centering
 \includegraphics[width=\textwidth]{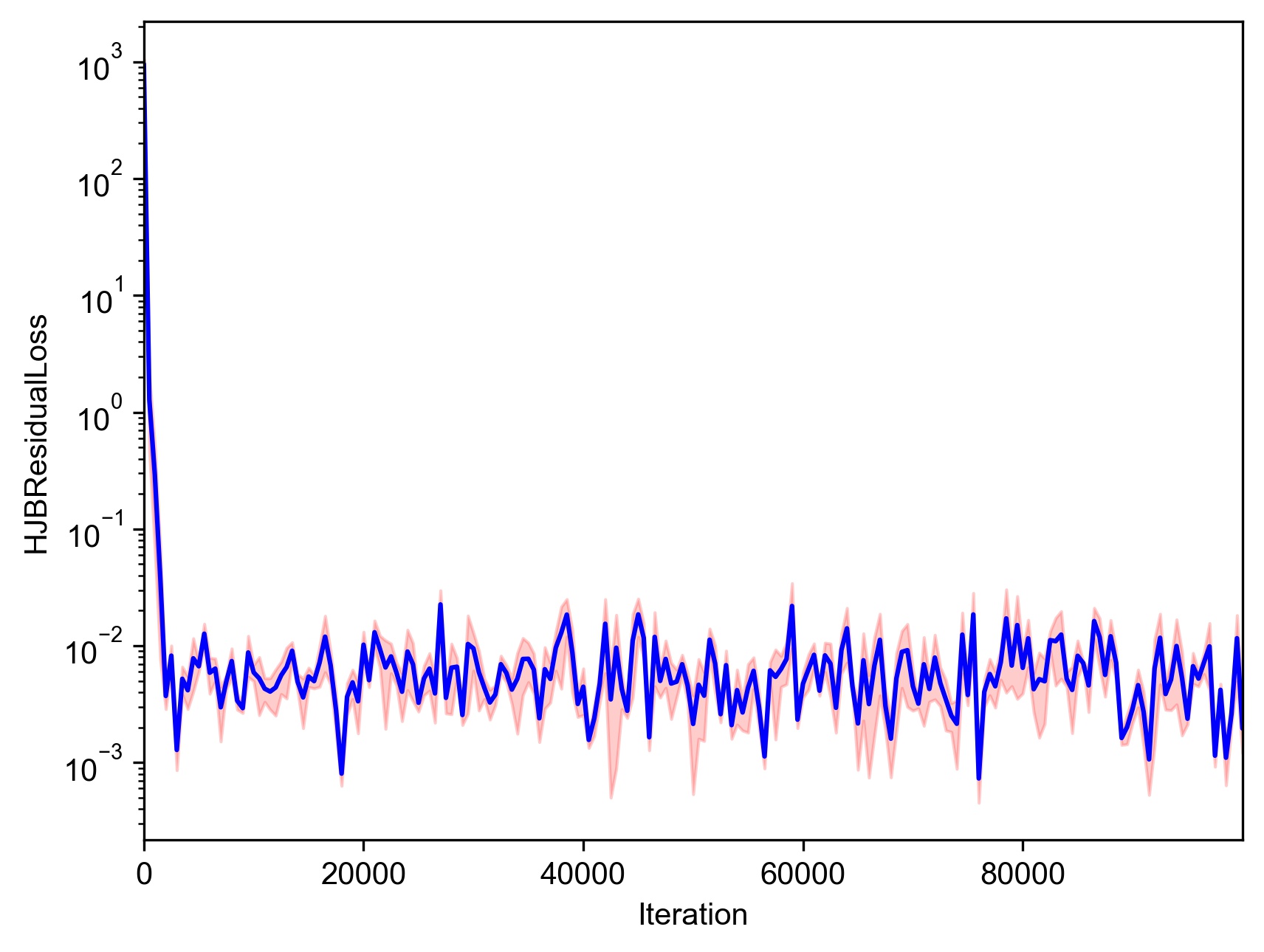}
 \caption{$B_g=B_d=256$.}
 \label{subfig:hjb-lgnum}
 \end{subfigure}
 \caption{HJB residual loss under different minibatch size.}
 \label{fig:hjb-num}
\end{figure}
\begin{figure}[!ht]
 \centering
 \begin{subfigure}[b]{0.4\columnwidth}
 \centering
 \includegraphics[width=\textwidth]{figures/ablation_study/FPResidualLoss.jpg}
 \caption{$B_g=B_d=32$.}
 \label{subfig:fp-smnum}
 \end{subfigure}
 \begin{subfigure}[b]{0.4\columnwidth}
 \centering
 \includegraphics[width=\textwidth]{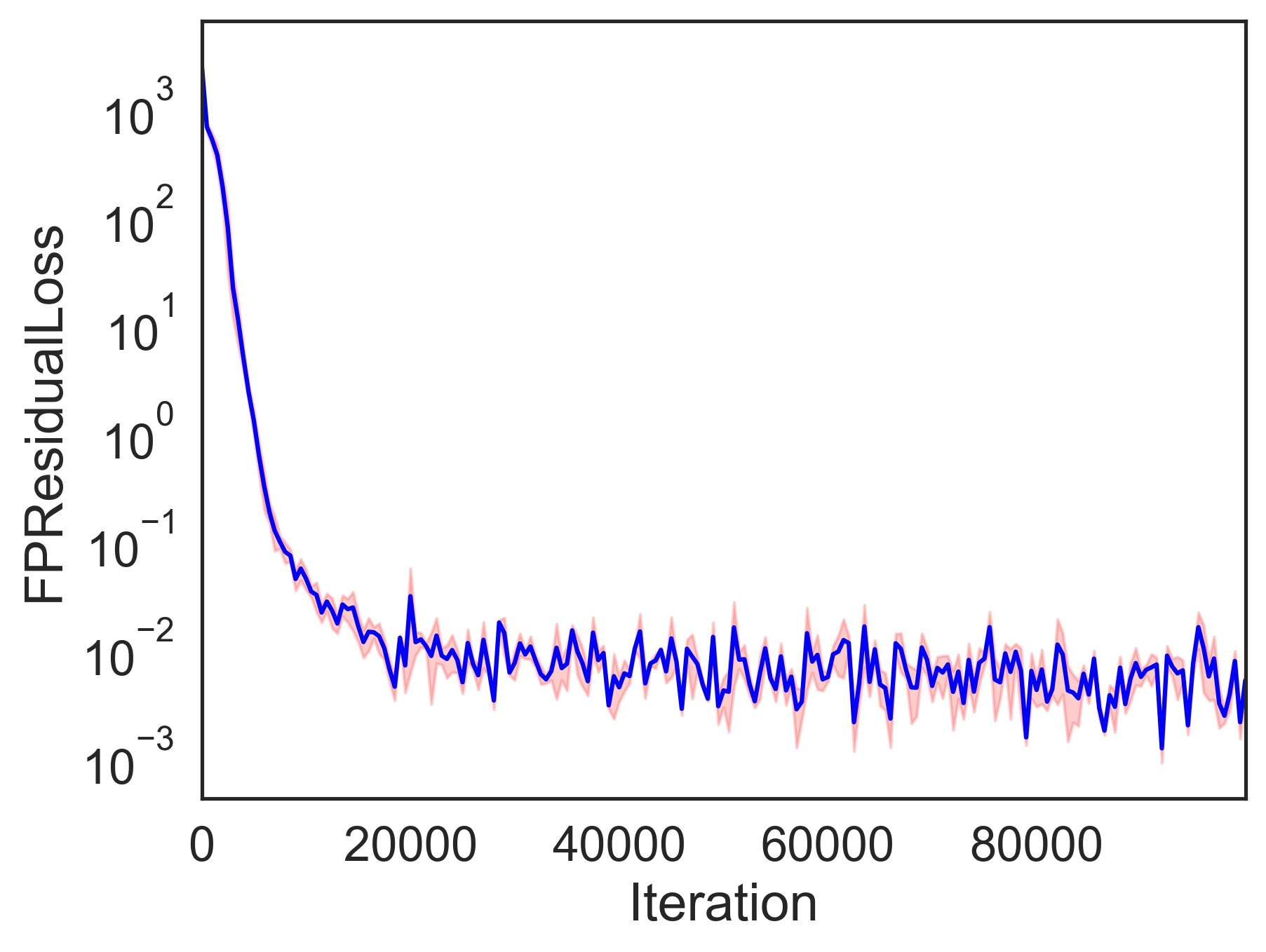}
 \caption{$B_g=B_d=256$.}
 \label{subfig:fp-lgnum}
 \end{subfigure}
 \caption{FP residual loss under different minibatch size.}
 \label{fig:fp-num}
\end{figure}

\paragraph{Sensitivity analysis.}
To understand possible contributing factors for the oscillation in the loss, especially for $u$, a sensitivity analysis on the learning rate of the generator $\alpha_g$ is conducted. In our test, the initial learning rate for the Adam Optimizer $\alpha_g^0$ takes the values of $5\times10^{-4}$, $1\times10^{-3}$, $1\times10^{-2}$ and $5\times10^{-2}$, respectively. 

In Figures \ref{subfig:1dim-rel-err-u} and \ref{subfig:1dim-rel-err-m}, the relative $l_2$ error on $u$ oscillates more than that of $m$. Similar phenomenon is observed in Figure \ref{fig:ab-lrg-rel-err}. In particular, in Figure \ref{subfig:ab-lrg-rel-err-u}, a drastic decrease in oscillation can be seen as the generator learning rate $\alpha_g$ decreases.

Turning to the differential residual losses, one can observe from Figures \ref{fig:hjb-lrg} and \ref{fig:fp-lrg} that decreasing $\alpha_g^0$ from $5\times10^{-2}$ to $5\times10^{-4}$
reduces the residual losses for both HJB and FP  with less oscillation.

Another parameter of interest is the number of samples in each minibatch, i.e., $B_g$ and $B_d$ in Algorithm \ref{alg:mfgan-dyn}. 
The cases of $B_g=B_d=32$, $64$, $128$, and $256$ are tested respectively.
Figure \ref{subfig:ab-num-rel-err-u} shows that the relative $l_2$ error of $u$ oscillates less as $B_g$ and $B_d$ increases from $32$ to $256$. 
Moreover, comparing the cases of $B_g=B_d=32$ and $B_g=B_d=256$, the residual losses for both HJB and FP decrease with less oscillation as minibatch size increases, as shown in Figures \ref{fig:hjb-num} and \ref{fig:fp-num}.

\begin{figure}[!ht]
 \centering
 \begin{subfigure}[b]{0.4\columnwidth}
 \centering
 \includegraphics[width=\textwidth]{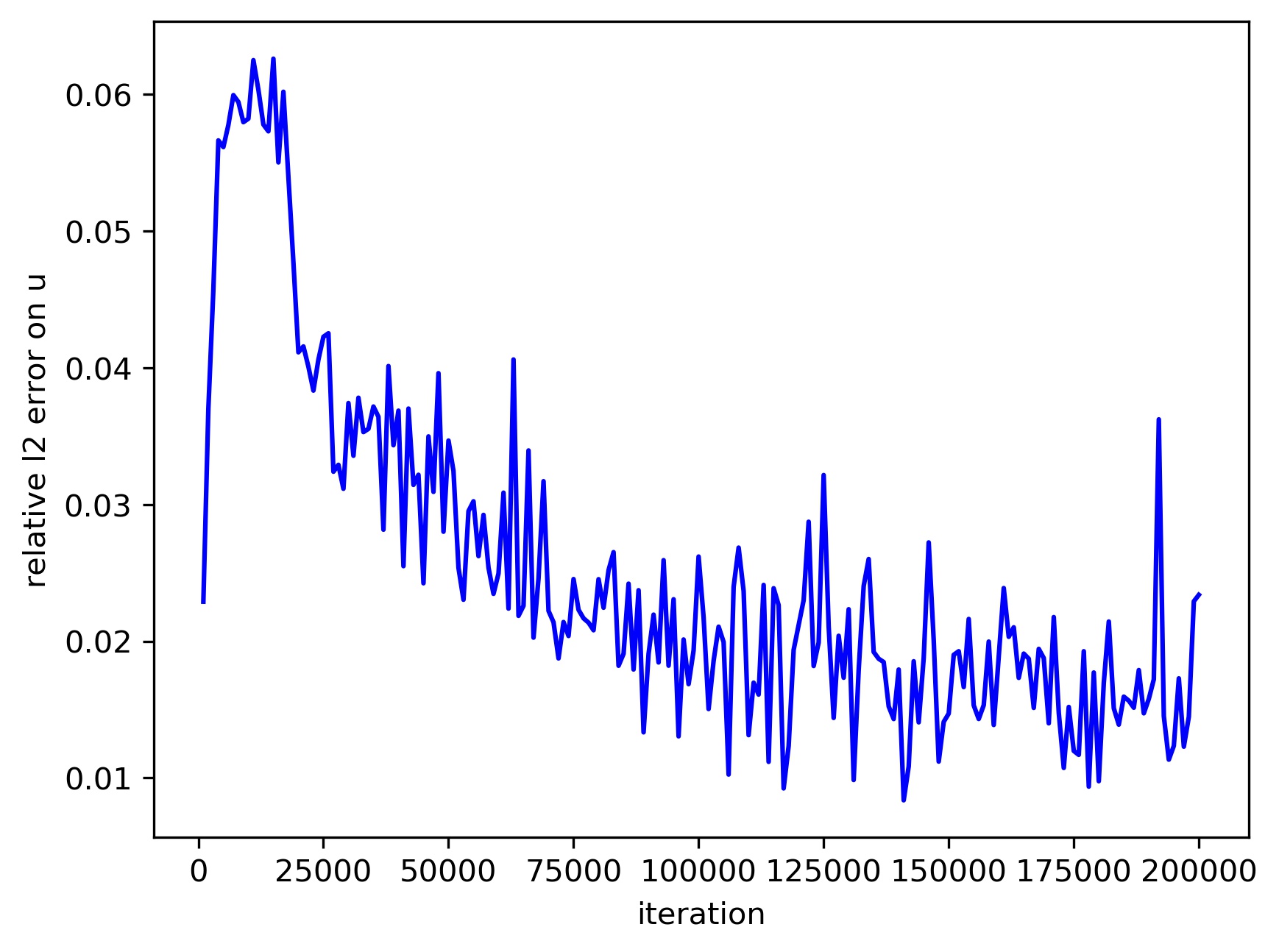}
 \caption{Relative $l_2$ error of $u$.}
 \label{subfig:4dim-rel-err-u}
 \end{subfigure}
 \begin{subfigure}[b]{0.4\columnwidth}
 \centering
 \includegraphics[width=\textwidth]{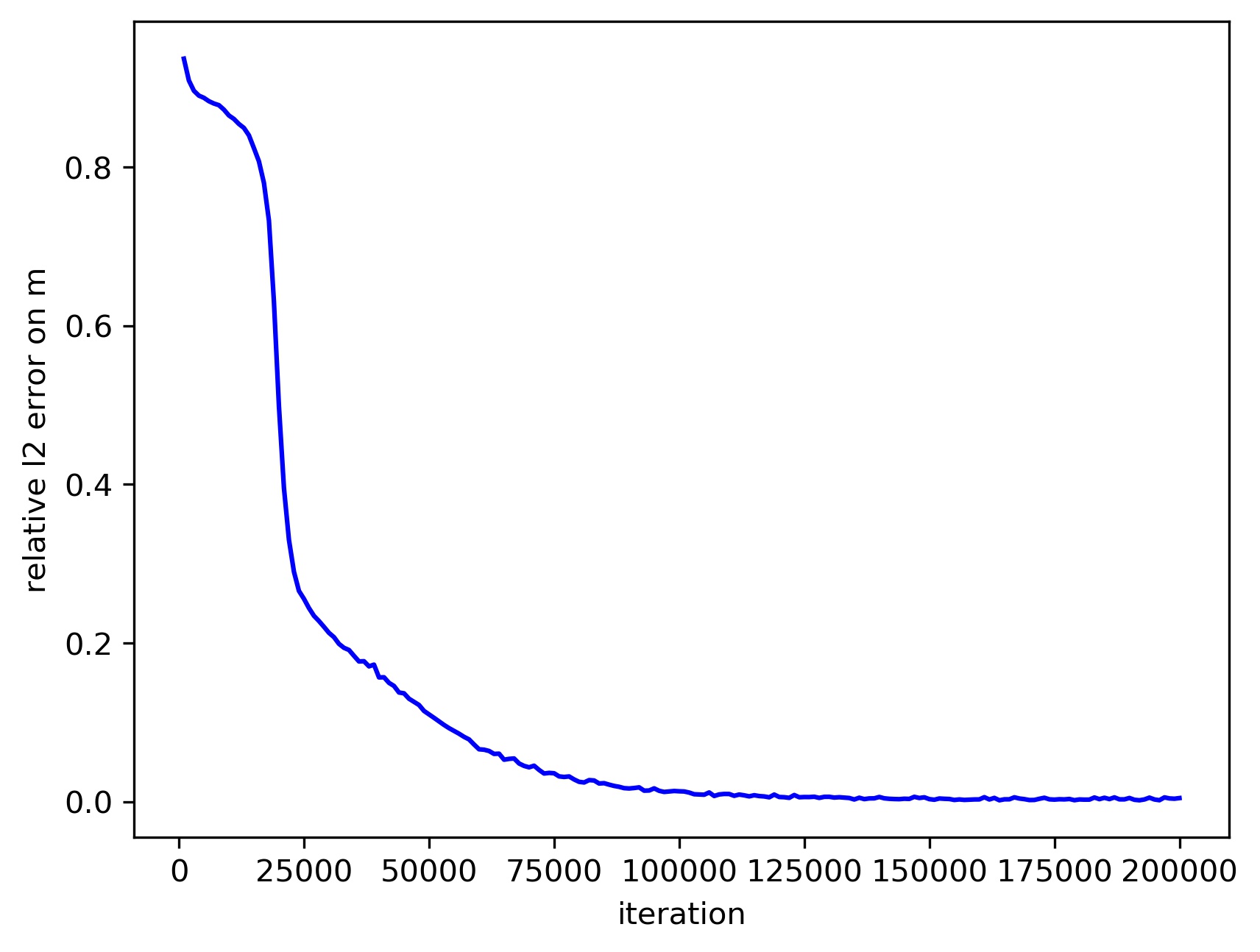}
 \caption{Relative $l_2$ error of $m$.}
 \label{subfig:4dim-rel-err-m}
 \end{subfigure}\\
 \begin{subfigure}[b]{0.4\columnwidth}
 \centering
 \includegraphics[width=\textwidth]{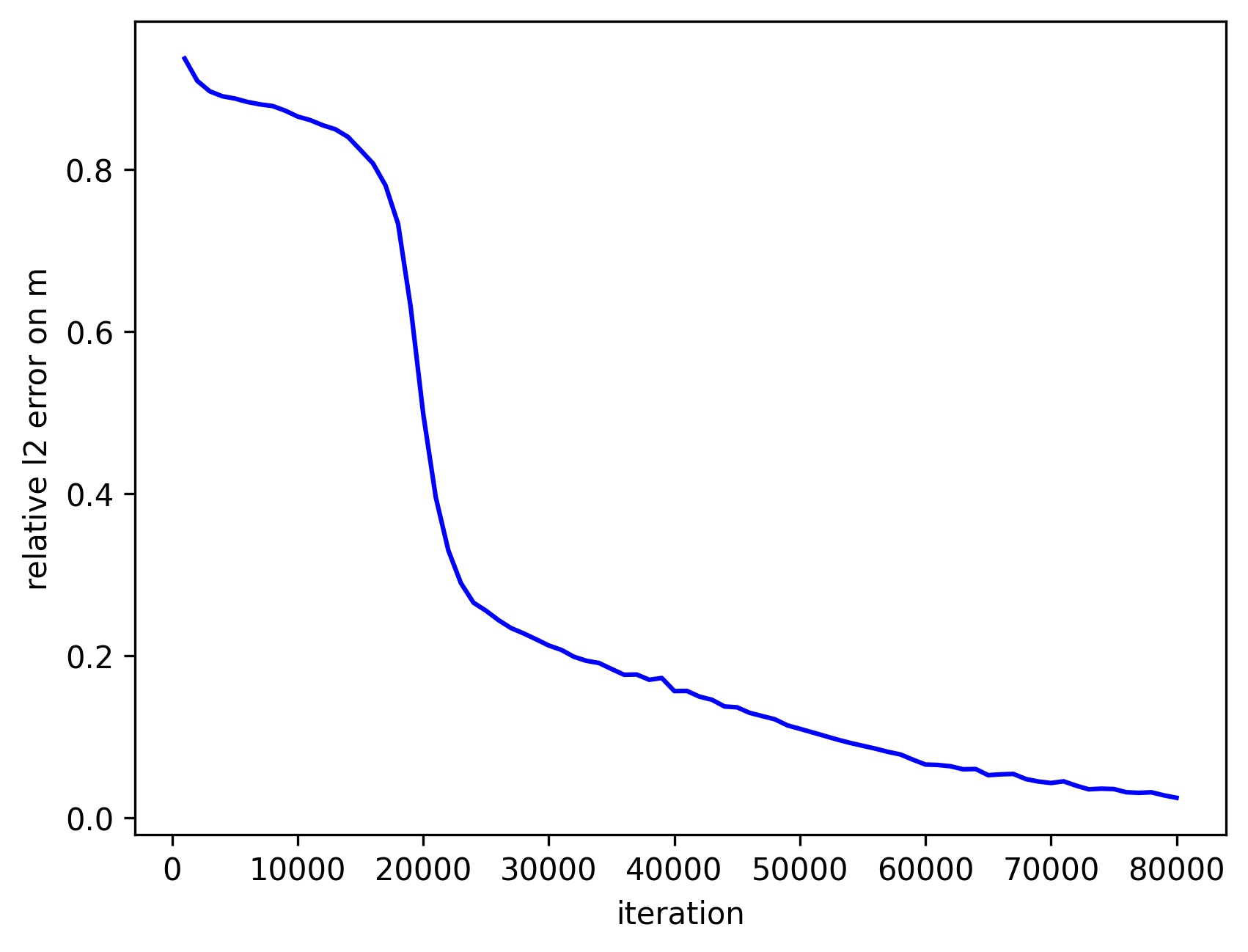}
 \caption{Error of $m$ first 80k iterations.}
 \label{subfig:4dim-rel-err-m-dec}
 \end{subfigure}
 \begin{subfigure}[b]{0.4\columnwidth}
 \centering
 \includegraphics[width=\textwidth]{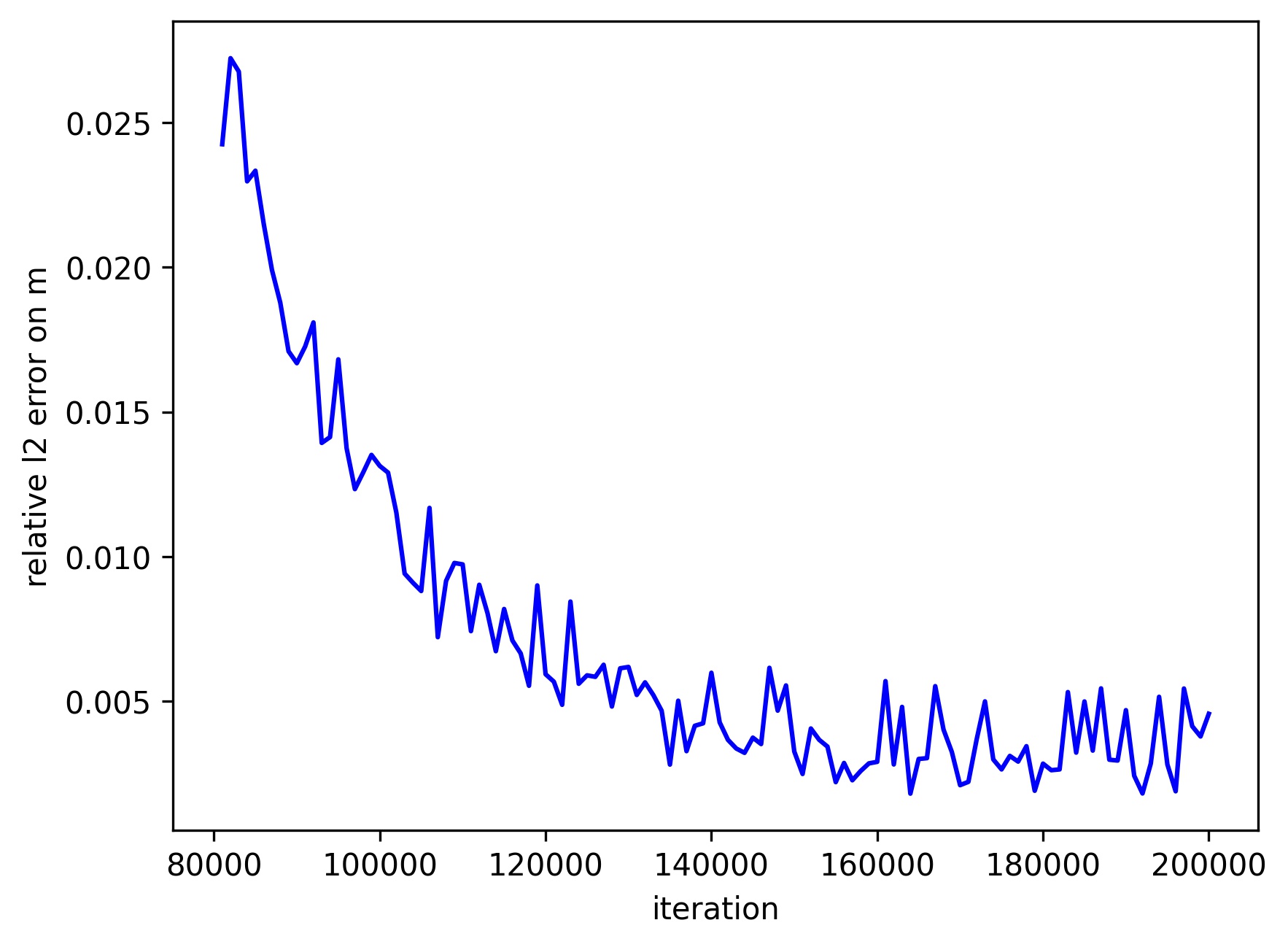}
 \caption{Error of $m$ after 80k iterations.}
 \label{subfig:4dim-rel-err-m-flt}
 \end{subfigure}
 \caption{Input of dimension 4.}
 \label{fig:4dim}
\end{figure}

\paragraph{Experiments in dimension 4. }
We also conduct experiments in dimension $4$. In this case, a different initial learning rate $5\times10^{-4}$ is adopted for the generator, and the number of outer loops is increased to $K=2\times 10^5$.

Figure \ref{fig:4dim} shows the relative $l_2$ errors between learnt and true value and density functions.
Within $2\times 10^5$ iterations, the relative $l_2$ error of $u$ decreases below $2\times10^{-2}$ and that of $m$ decreases to $4\times10^{-3}$. In fact, the training result stabilizes after around $8\times10^4$ iterations.

Finally, it is worth noting that similar experiment for dimension $4$ has been conducted in \cite{CarmonaLauriere_DL} with a simple fully connected architecture; see their Test Case 4. In comparison, their algorithm needs a larger number of iterations ($10^6$ of iterations vs our $8\times10^4$), in order to achieve the same level of accuracy. The computational cost of one SGD iteration with fully connected architecture is lower than with the DGM architecture that we use but making a rigorous comparison in terms of computational time is quite challenging due to the details of the implementation. However, needing a smaller number iterations (and hence a smaller number of training samples) is an objective advantage of the results we obtained here.

\subsubsection{Ergodic MFGs without explicit solutions}
We next test Algorithm \ref{alg:mfgan-dyn} for a class of MFGs for which no explicit solutions are available. Take ergodic MFGs \eqref{eq:cost} on $\mathbb{R}^d$, where
\[L(x,\alpha) = \frac{1}{2} |\alpha|^2 + \tilde f(x),
	\quad
	f(x, m) = m^2+1,\quad \epsilon=\frac{1}{2},\]
with
\[\tilde f(x)= \frac{1}{2}\sum_{i=1}^d\left[\sin(2\pi x_i)+\cos(2\pi x_i)\right].\]
The solution to this MFG  can be characterized  by the coupled PDE system \eqref{eq:hjb-fp}. 

In the experiment, we will show that  Algorithm \ref{alg:mfgan-dyn} not only can be adopted to solve this MFG problem in a multi-dimensional case, for instance $d=8$, but also  can capture the periodicity with appropriate choices of error functions, even without the {\it a priori} transformation of input given by \eqref{eq:input-fourier}.

\paragraph{Experiment setup.}
Similar as the previous experiments, the DGM model is adopted to parametrize both the value function $u$ and the density function $m$ with parameters $\theta$ and $\omega$, respectively and a maximum entropy probability distribution is adopted when modelling the density function $m$. There are a few differences compared with the previous experiment setup as follows.
\begin{itemize}
    \item In the DGM network to parametrize the density function, the activation function of the output layer of $m_\omega$ is set specifically to be exponential. Instead of normalizing $m_{\omega}$ beforehand, we utilize the following normalization condition for the density function
    \[\hat L_{MF,norm}=\left[\frac{\sum_{i=1}^{B_d}m_\omega(x_i)}{B_d}-1\right]^2,\]
    with $\beta_{MF}=10^3$.
    \item Both DGM networks of $u_\theta$ and $m_\omega$ contain $3$ hidden layers with $40$ nodes, with hyperbolic tangent as activation functions.
    \item Within each iteration of training, the SGD steps for updating $\theta$ and $\omega$ are now set to be $N_{\theta}=N_{\omega}=10$. Initial learning rates for both neural networks are chosen to be $3\times10^{-3}$. The minibatch sizes are $B_g=B_d=1000$.  The total number of iterations is $K=10^5$.
    \item To capture the periodicity condition for both $u_\theta$ and $m_\omega$, additional penalty terms are added to both $\hat L_{MF}$ and $\hat L_{Val}$ such that
    \[\hat L_{Val,\,per}=\sum_{i=1}^d\sum_{(z^1,z^2)\in \mathcal{B}_i}[u_\theta(z^1)-u_\theta(z^2)]^2, \ \ \hat L_{MF,\,per}=\sum_{i=1}^d\sum_{(z^1,z^2)\in \mathcal{B}_i}[m_\omega(z^1)-m_\omega(z^2)]^2,\]
    with $\mathcal{B}_i=\{((x_1,\dots,x_d),(x_1,\dots,x_{i-1},1-x_i, x_{i+1},\dots, x_d)) :x_j\in\{0,1\}\ \ j=1,\dots,d\}$. The penalty coefficients are set to be 10 for both $\hat L_{Val,\,per}$ and $\hat L_{MF,\,per}$.
\end{itemize}

\begin{figure}
    \centering
    \begin{subfigure}[b]{0.45\columnwidth}
        \centering
        \includegraphics[width=\textwidth]{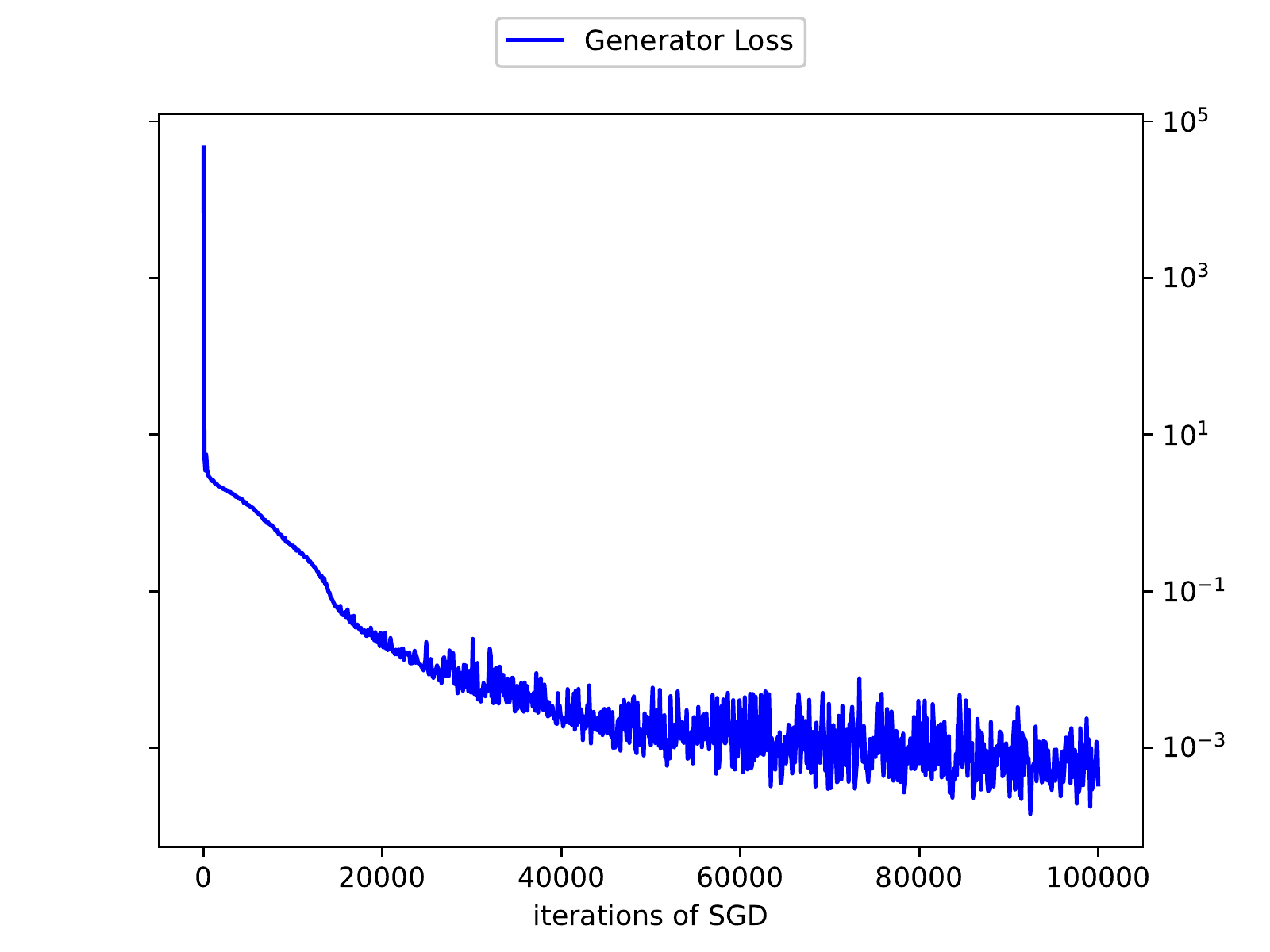}
        \caption{Generator loss}
        \label{subfig:8d-gen-loss}
    \end{subfigure}
    \begin{subfigure}[b]{0.45\columnwidth}
        \centering
        \includegraphics[width=\textwidth]{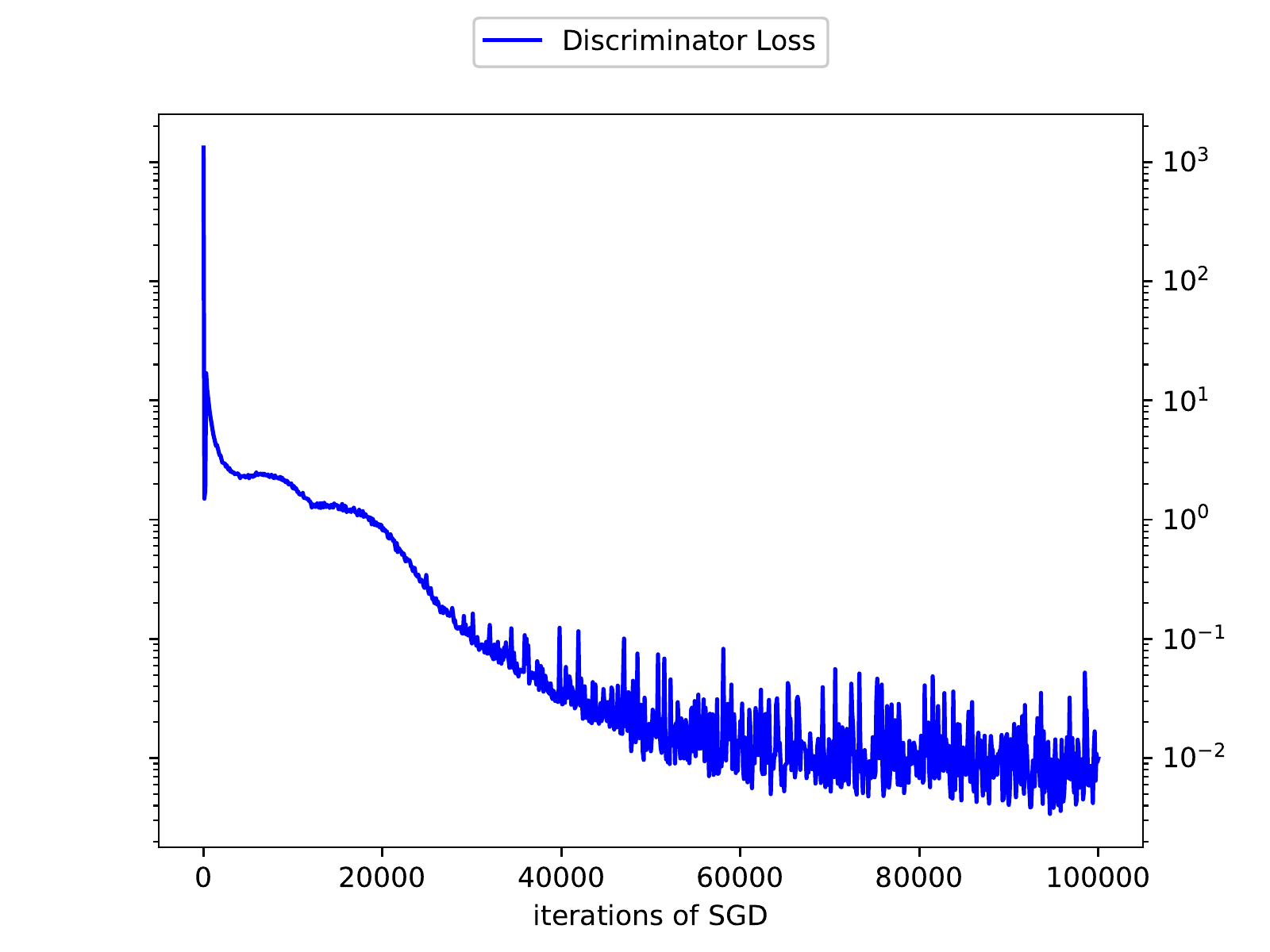}
        \caption{Discriminator loss}
        \label{subfig:8d-disc-loss}
    \end{subfigure}
    \caption{Training losses}
    \label{fig:8dim}
\end{figure}
The results are summarized in Figure \ref{fig:8dim}. After $10^5$ iterations of training, the generator loss $\hat L_{Val}$ drops below $10^{-3}$ and the discriminator loss $\hat L_{MF}$ reaches $10^{-2}$.
\begin{figure}[!ht]
    \centering
    \begin{subfigure}[b]{0.45\columnwidth}
        \centering
        \includegraphics[width=\textwidth]{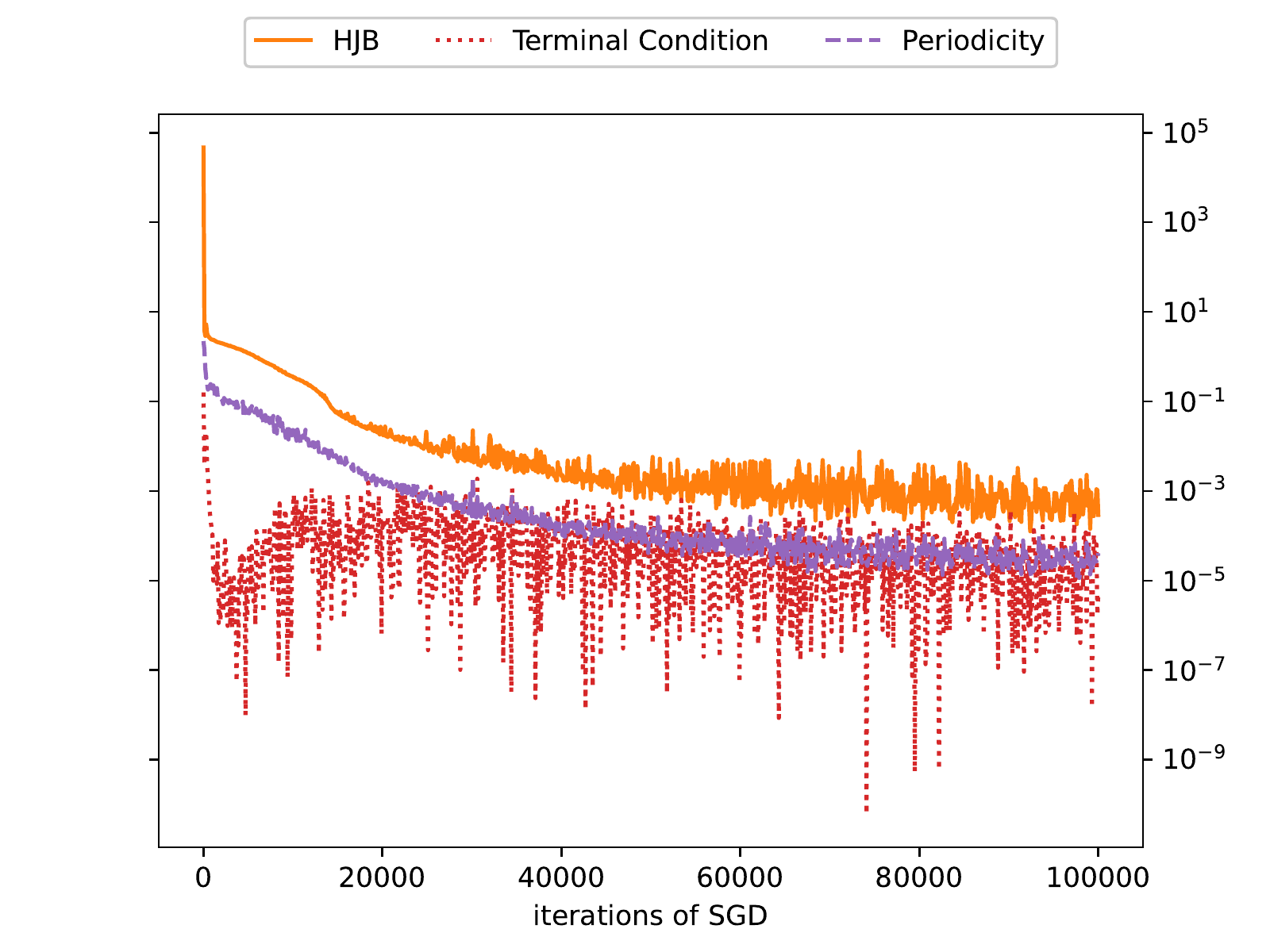}
        \caption{Generator detailed losses}
        \label{subfig:8d-gen-loss-detailed}
    \end{subfigure}
    \begin{subfigure}[b]{0.45\columnwidth}
        \centering
        \includegraphics[width=\textwidth]{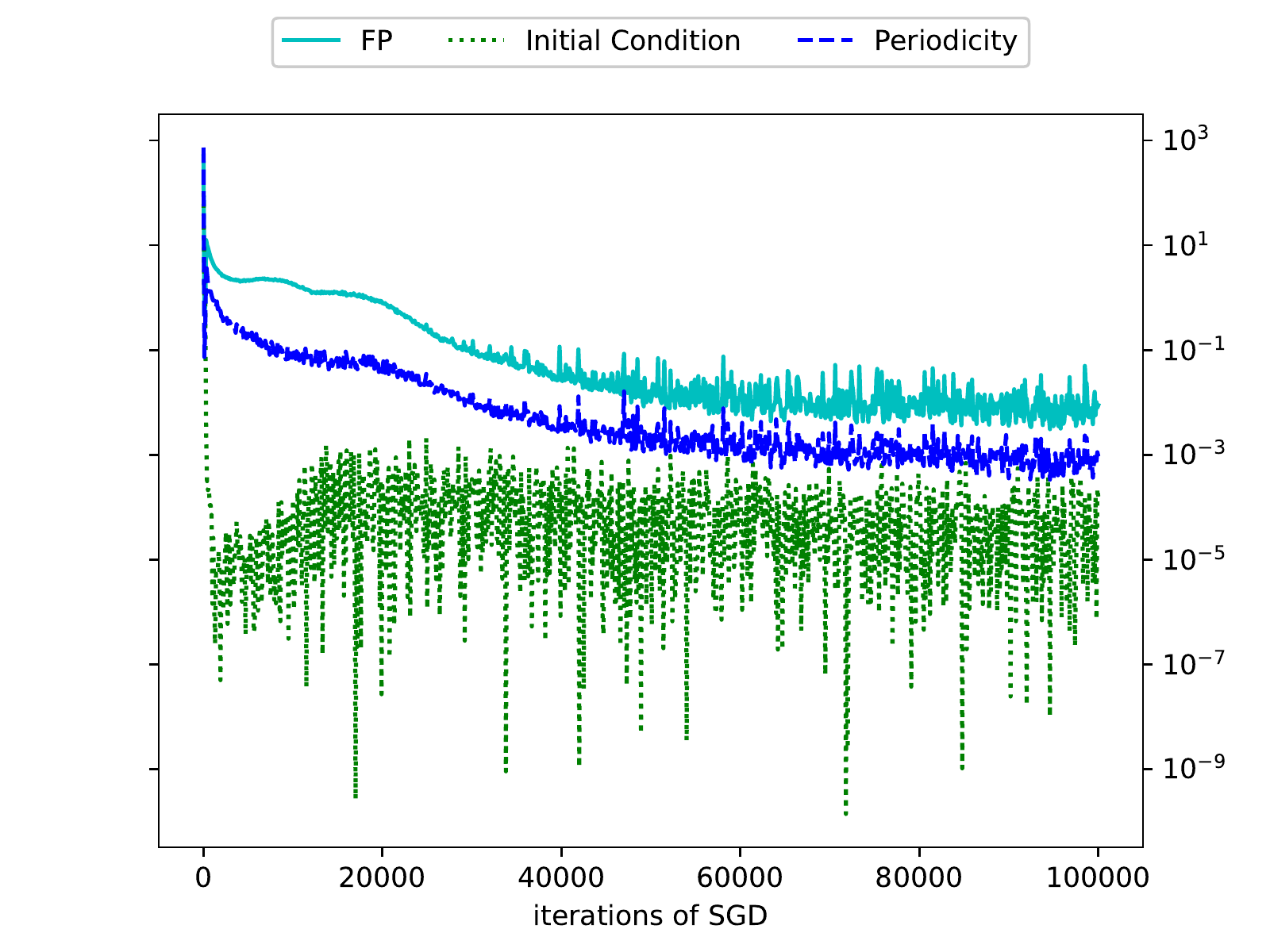}
        \caption{Discriminator detailed losses}
        \label{subfig:8d-disc-loss-detailed}
    \end{subfigure}
    \caption{Detailed training losses}
    \label{fig:8dim-detail}
\end{figure}
The detailed losses are shown in Figure \ref{fig:8dim-detail}. Note that not only the coupled PDE system \eqref{eq:hjb-fp} is numerically solved accurately, the periodicity condition for both the value function and the density function has been well satisfied.
\bibliography{gan-mfg}
\bibliographystyle{plain}

\end{document}